\documentclass[a4paper,UKenglish,cleveref, autoref, thm-restate]{lipics-v2021}

\usepackage{tikz}
\usepackage{pgfplots}
\usetikzlibrary{automata, arrows.meta, positioning,decorations.pathmorphing}
\usepackage[ruled,vlined]{algorithm2e}
\DontPrintSemicolon
\LinesNumbered

\usepackage{thm-restate}
\usepackage{marginnote}
\newcommand{\reviewer}[1]{ \marginpar{\it\footnotesize #1}}

\usepackage{wrapfig}
\usepackage{lipsum}

\usepackage{nicefrac}
\usepackage{xcolor}

\nolinenumbers

\title{A Fine-Grained Analysis of SIMD-Based Hash Table Implementations}

\titlerunning{A Fine-Grained Analysis of SIMD-Based Hash Table Implementations} 

\author{Cyril Nicaud}{Univ Gustave Eiffel, CNRS, LIGM, F-77454 Marne-la-Vallée, France}{cyril.nicaud@univ-eiffel.fr}{}{}

\author{Pablo Rotondo}{Univ Gustave Eiffel, CNRS, LIGM, F-77454 Marne-la-Vallée, France}{pablo.rotondo@univ-eiffel.fr}{}{}

\authorrunning{C. Nicaud and P. Rotondo} 

\Copyright{Cyril Nicaud and Pablo Rotondo} 



\category{} 

\relatedversion{} 

\newcommand{\bucketSize}{b}

\newcommand{\finger}{h_f}
\newcommand{\overflow}{h_o}
\renewcommand{\O}{\mathcal{O}}
\newcommand{\instruction}[1]{\texttt{#1}}

\begin{document}

\maketitle

\begin{abstract}
In recent years, several variants of classical hash table schemes have been developed by engineers in order to take advantage of the processor’s internal parallelism using SIMD instructions, which make it possible to operate on multiple bytes simultaneously. At a small additional memory cost, this enables a significant speedup, making it a data structure that is increasingly popular in practice, when very high performance is required. 
In this article, we provide a detailed theoretical analysis of the dynamics of such hash tables. From a methodological standpoint, we use and adapt a technique developed by Wormald in the 1990s to study dynamic graphs. This approach, which can be  adapted to many variants, enables us to accurately estimate the quantities of interest by capturing the dynamics of the data structure through systems of differential equations. Although complex, we provide an explicit description of the solutions of these systems, which can furthermore be efficiently approximated numerically. Our main results are stated with high probability, which is significantly more precise than average-case analyses, and they match experimental results remarkably well, even for hash tables of moderate size.

\end{abstract}

\newpage

\section{Introduction}
\label{sec:intro}

Hash tables are a fundamental data structure that has become indispensable in practice and is now natively available in nearly all programming languages, yet their concrete implementation varies considerably from one language to another. For instance, one can find separate chaining (\textsc{Java}, \textsc{C++}), separate chaining with array-backed nodes (\textsc{C\#}), open addressing with various probing schemes: pseudo-random perturbation (\textsc{Python}, \textsc{Ruby}), quadratic probing (\textsc{V8}), linear probing (\textsc{Swift}), \ldots  

The importance of hash tables is such that the organization of the table structures themselves takes up a non-negligible portion of the RAM of data servers and search engines, which has led engineers developing programming languages and performance-oriented libraries to revisit the classic academic results on hash tables in order to optimize them in terms of space and time. While the classic theoretical literature~\cite{guibas1978,knuth1998art,lueker1993}  long favored minimizing collisions as the primary performance criterion, a broader trend has emerged in recent years: the dominant constraint on modern hardware is cache locality rather than the number of collisions in the combinatorial sense, which has driven the convergence of several major languages and libraries toward the \emph{Swiss Table} design, as seen in Google's Abseil (C++), Rust, Go, and Meta's F14. They all use compact metadata and SIMD instructions to drastically reduce memory accesses per probe. The goal of this paper is to propose a mathematical methodology for analyzing this type of hash table.

\paragraph*{Model of hash tables}
Open addressing is by far the most widely used method nowadays when it comes to optimizing performance. With this method, collisions are resolved through probing: in case of collisions when inserting a key, several strategies exist to find a free slot, such as linear probing, where the key is inserted into the next available slot, quadratic probing, double hashing, \ldots This is usually preferred  to closed addressing, where each slot in the table can hold any number of keys, in a linked list, for instance. This choice minimizes memory overhead and takes cache locality into account, at the cost of slightly more algorithmic management when collisions occur. Above all, it is compatible with the use of SIMD instructions, which make it possible to perform several elementary operations simultaneously within the processor. Thus, if we look in detail at hash tables in {\tt Rust}, in the {\tt boost} library of {\tt C++}, or in the  implementations used by \textsc{Google} (Abseil) or \textsc{Meta} (F14), we see these two characteristics every time: open addressing and vectorization with SIMD instructions.

 The mathematical analysis of hash tables, as initiated by D. Knuth~\cite{knuth1998art}, we avoid complex considerations about hash functions by adopting the \emph{uniform hashing model}, where each hash function $h$ is modeled by a uniform random map: if the table has size $n$, the image of any key $x$ by $h$ is taken uniformly at random and independently in $\{0,\ldots,n-1\}$. When several hash functions are used, they are independent (in the probabilistic sense).

 For collision resolution, we use the \emph{random probing} model, where we pick a slot uniformly at random each time (which amounts to using a new hash function at each iteration of the probe sequence). In practice, setting aside linear probing which is very different, quadratic probing is one of the most used. Experimentally, the performance of quadratic probing does not differ very significantly from that of random probing~\cite{knuth1998art}. However, the theoretical analysis of quadratic probing is particularly difficult, and no significant result existed in the literature before the recent breakthrough by Kuszmaul and Xi~\cite{kuszmaul2024}.

Finally, from reading engineers’ blog posts\footnote{For instance see \href{https://bannalia.blogspot.com/2022/11/inside-boostunorderedflatmap.html}{\color{blue}{Joaqu\'in L\'opez Mu\~noz' post}} regarding {\tt boost::unordered\_flat\_map}. } and from the conferences they gave about their implementations, the main performance metric is the number of hops (i.e., the number of probes to find a free slot), since this is quite directly related to cache misses, the cost of other operations being of secondary importance. Consequently, in this article we study hash tables in which collisions are handled by open addressing, under the uniform hashing model and random probing, where the time cost is estimated by counting the number of hops.

It is standard practice not to specify the size of a hash table when it is created, and to resize it,  reinserting all its elements in the process, when it becomes too full. We do not include this mechanism in our model, since we are primarily interested in the state of the table after a given number insertions, which is exactly equivalent to looking at the situation just after the last resizing. This changes the total construction cost, but it is easy to derive it from our results using a standard amortized analysis once  the resizing parameters are given.

\paragraph*{Using vectorization for hash tables operations}
SIMD instructions introduce parallelism within the CPU by allowing simultaneous operations on multiple bytes (typically 16 bytes in one instruction on a modern architecture). 

In a nutshell, the general mechanism considered in this article can be described as follows.  The table $T$ is made of $n$ buckets, each containing up to $b$ keys (for example, $b=16$ to match the inner parallelism of the processor). The table is accompanied by a header $H$, also composed of $n$ buckets, each containing $b$ bytes. When a key $x$ is inserted into the hash table, it is assigned to a bucket $T[i]$ identified through a probing sequence that searches for a bucket that is not full. We also use a hash function $\finger$ that maps $x$ to a byte, say, an integer between $0$ and $254$, to add $\finger(x)$ to the header, in the bucket $H[i]$ (see Fig.~\ref{fig:SIMD}). The value $255$ (or equivalently $-1$) is used to denote an empty slot. This byte $\finger(x)$ is used as a \emph{fingerprint} for $x$. When searching for $x$ in the hash table, each probe (at position $j$) begins by using a few SIMD instructions to quickly test whether $\finger(x)$ appears in $H[j]$. If the fingerprint is absent, the algorithm can skip querying $T[j]$ altogether. It is only when we try to insert into a bucket that already contains 16 keys that we need to perform a hop to try a new slot. 

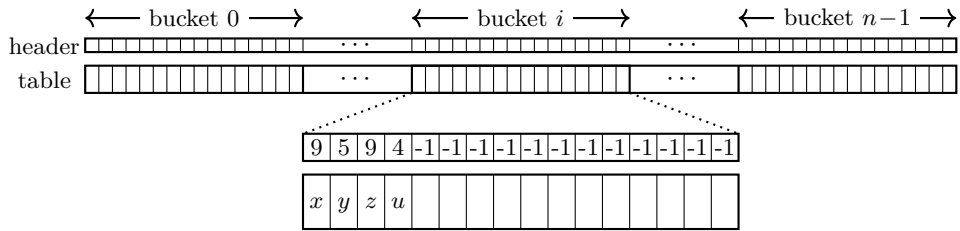
\begin{figure}[h]
	\begin{tikzpicture}[scale=.18]
		\draw[<->,thick] (0,5.5) -- node[fill=white]{\small bucket $0$} (16,5.5);
		\draw[<->,thick] (48,5.5) -- node[fill=white]{\small bucket $n\!-\!1$} (64,5.5);
		\draw[<->,thick] (24,5.5) -- node[fill=white]{\small bucket $i$} (40,5.5);
	
		\draw[thick] (0,0) rectangle (64,2);
		\foreach \x in {0,1,...,15} {
			\draw (\x,0) -- (\x,2);
			\draw (48+\x,0) -- (48+\x,2);
			\draw (24+\x,0) -- (24+\x,2);
		}
		\draw[thick] (16,0) rectangle (16,2);
		\draw[thick] (48,0) rectangle (48,2);	
		\draw[thick] (24,0) rectangle (40,2);
		
		\node at (20,1) {$\cdots$};
		\node at (44,1) {$\cdots$};

		\draw[thick] (0,3) rectangle (64,4);
		\foreach \x in {0,1,...,16} {
			\draw (\x,3) -- (\x,4);
			\draw (48+\x,3) -- (48+\x,4);
			\draw (24+\x,3) -- (24+\x,4);
		}
		\node at (20,3.5) {$\cdots$};
		\node at (44,3.5) {$\cdots$};

		\node at (-3,3.5) {\small header};
		\node at (-3,1) {\small table};
		
		\draw[dotted, thick] (24,0) -- (16,-3);
		\draw[dotted, thick] (40,0) -- (48,-3);
		
		\draw[thick] (16,-3) rectangle (48,-5);
		\foreach \x in {0,1,...,15} {
			\draw (2*\x+16,-3) -- (16+2*\x,-5);
		}
		\draw[thick] (16,-6) rectangle (48,-10);
		\foreach \x in {0,1,...,15} {
			\draw (2*\x+16,-6) -- (16+2*\x,-10);
		}

		\foreach \x/\h/\k in {0/9/x,1/5/y,2/9/z,3/4/u} {
			\node at (17+2*\x,-4) {\small \h};
			\node[anchor=north] at (17+2*\x,-7) {\small $\k$};
		}
		\foreach \x in {4,5,...,15}{
			\node at (17+2*\x,-4) {\small -1};
		}
	\end{tikzpicture}
    \caption{The hash table for $b=16$. The $n$ buckets are made of a size-16 header to store the fingerprints and a size-16 array to store the keys. In this example, the $i$-th bucket contains the keys $x$, $y$, $z$ and $u$, of fingerprints $\finger(x)=\finger(z)=9$, $\finger(y)=5$ and $\finger(u)=4$. When looking for a key $k$ in this bucket, the algorithm first computes $\finger(k)$ and then uses a few SIMD instructions to check whether $H[i]$ contains this fingerprint, in order to know whether it has to look in $T[i]$ or not.  
    \label{fig:SIMD}}
\end{figure}

In this article, we do not mathematically study the impact of deletions, mainly due to the lack of a simple, commonly accepted scenario for hash table usage that is relevant in practice (see the end of the conclusion for a discussion of this point and possible directions for extension). For the sake of completeness, we nevertheless present the algorithmic mechanisms that make it possible to handle deletions for these hash tables.
Deletions are handled using \emph{tombstones}: another of the 256 header values is used to mark bucket slots that are unused but previously contained a key that has been deleted. For instance, $\finger$ maps a key to $\{0,\ldots, 253\}$, the value $255\equiv-1$ denote an empty slot and the value $254\equiv-2$ denote a tombstone. A bucket has an available slot for inserting a new key when its header  contains at least a $-1$ or a $-2$ and requires a hop for lookups when the current search is unsuccessful and its header does not contain a $-1$.

\begin{figure}[h]
\centering
\begin{tikzpicture}[scale=.18]
    \draw[thick] (31,3) rectangle (33,5);
    \node at (32,4) {\small 4};
    \node at (28,4) {\small $\finger(x)$};

    \draw[->,thick] (32,2.5) -- node[right]{\  \instruction{\_mm\_set1\_epi8}} (32,0.5);
    \foreach \x in {0,1,...,15} {
        \node at (17+2*\x,-1) {\small 4};
    }

    \draw[thick] (16,0) rectangle (48,-2);
    \node at (13,-4) {\small $H[i]$};
    \draw[thick] (16,-3) rectangle (48,-5);
	\foreach \x in {0,1,...,15} {
        \draw (2*\x+16,0) -- (16+2*\x,-2);
        \draw (2*\x+16,-3) -- (16+2*\x,-5);
	}
	\foreach \x/\h in {0/9,1/5,2/-2,3/4,4/7,5/4} {
			\node at (17+2*\x,-4) {\small \h};
	}
	\foreach \x in {6,7,...,15}{
			\node at (17+2*\x,-4) {\small -1};
	}

    \draw[->,thick] (32,-5.5) -- node[right]{\  \instruction{\_mm\_cmpeq\_epi8}} (32,-7.5);

    \draw[thick] (16,-8) rectangle (48,-10);
    \foreach \x in {0,1,...,15} {
        \draw (2*\x+16,-8) -- (16+2*\x,-10);
    }
    \foreach \x/\h in {0/0,1/0,2/0,3/-1,4/0,5/-1} {
			\node at (17+2*\x,-9) {\small \h};
    }
    \foreach \x in {6,7,...,15}{
			\node at (17+2*\x,-9) {\small 0};
	}
\end{tikzpicture}
\caption{
The main SIMD instructions at play, without going into all the details. We use  x86-64 assembly for the instruction's names. When looking for a key $x$, the algorithm first computes $\finger(x)$, then creates a 128-bit SIMD integer vector made of 16 copies of the byte $\finger(x)$ using the instruction \instruction{\_mm\_set1\_epi8}. Each time we search for $x$ in a bucket $T[i]$, we use the \instruction{\_mm\_cmpeq\_epi8} instruction to compare, in parallel, the 16 bytes of its header $H[i]$ with $\finger(x)$. If there is no match, the output is a vector of $0$'s, which can be checked very quickly. Otherwise, $x$ is compared to each key within the bucket having same fingerprint (yielding two comparisons in the example).
\label{fig:SIMD details}
}
\end{figure}
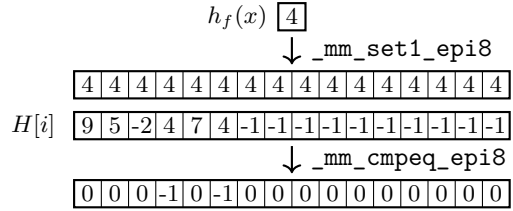

Fig.~\ref{fig:SIMD details} details the main mechanisms that allow quick probing of the $i$-th bucket. Throughout the article, the number of such probes is the main measure of complexity for our algorithms. Even though these queries do not always require the same amount of computation, the table jumps are the most costly, due to their impact on the cache, and are therefore the only parameter that really matters in practice. As a result, it will be important to know how many buckets are full, and to be able to calculate this, we will need to provide a precise estimation, at any given time, of how many buckets contain $j$ keys, for all  $j$ from $0$ to $16$.

Since we use headers, another idea that can be easily implemented is to store other types of information in them besides fingerprints. This can originally be attributed to Amble and Knuth~\cite{AmbleKnuth74}, who proposed a mechanism related to adding an extra bit in each slot of a classical open-addressing hash table. This bit is used as a flag to remember whether this slot has already been subject to a hop during an insertion. This costs some space, but it makes it possible to prematurely interrupt the probing sequence during a search. In our settings, this simple idea has been implemented in different ways, using one or more slots in the header to store some information. A common use is as follows: each bucket now contains up to $b=15$ keys, and the associated  header contains $15$ fingerprints and one byte to account for \emph{overflows}, i.e. when a hop occurred previously.  For the $i$-th probe in the sequence, we have a hash function $\overflow^{(i)}$ that returns a number between $0$ and $7$, and if we try to insert a key $x$ into a bucket that is already full, the $\overflow^{(i)}(x)$-th bit of the overflow byte is set to $1$ (initially, all its bits are set to $0$). 
This allows us, when searching for the key $x$, to avoid performing a hop if the overflow bit $\overflow^{(i)}(x)$ is zero, even if the bucket is full with $15$ keys. This only improves unsuccessful searches, but can be significantly faster. Using this method, one does not need tombstones anymore: one keeps looking for a key $x$ if its overflow bit $\overflow^{(i)}(x)$ is 1, even if the bucket is not full. Indeed, it indicates that the bucket was previously full and that $x$ may have been inserted elsewhere. If $\overflow^{(i)}(x)=0$ and $x$ is not in the current bucket, then $x$ cannot be in the table. Using a different overflow hash for each probe best describes what is observed in practice and is consistent with the model of random probing.

In the examples we found in different programming languages, we observed buckets of sizes between 14 and 16, and 0, 1, 7 or 8 overflow bits. In the sequel, these two values will be treated as parameters ($b$ for the buckets' size and $d$ for the number of overflow bits), allowing us to cover all cases.

\begin{figure}[t]
\noindent
{\small
\begin{minipage}{.45\textwidth}
\begin{algorithm}[H]
\caption{\!\!\textsc{AddNewElement}($x$)}
$c,i\gets 0,h^{(0)}(x)$\;
\While{$|T[i]|=b$}{
    $B[i][h_o^{(c)}(x)]\gets 1$\;
    $c\gets c+1$\;
    $i\gets h^{(c)}(x)$\;
}
Add $x$ in $T[i]$ and update $H[i]$\;
\end{algorithm}
\end{minipage}
\hfill
\begin{minipage}{.5\textwidth}
\begin{algorithm}[H]
\caption{\textsc{Search}($x$)}
$c,i\gets 0,h^{(0)}(x)$\;
\While{\emph{True}}{
\lIf{$x\in T[i]$}{\Return True}
\lIf{$B[i][h_o^{(c)}(x)]=0$}{\Return False}
    $c\gets c+1$\;
    $i\gets h^{(c)}(x)$\;
}
\end{algorithm}
\end{minipage}
}
\caption{The algorithms to insert a new element and to search for an element. We used $T$, $H$ and $B$ for the tables of buckets, headers and overflow bits, respectively. The SIMD operations are performed in Line~3 of \textsc{Search}. The random probing model is realized using sequences of hash functions $(h^{(c)})_{c\geq 0}$ for probing and $(\overflow^{(c)})_{c\geq 0}$ for overflows. We assume that the resizing is handled separately.\label{fig:algos}}
\end{figure}

\paragraph*{Related works}
Since Knuth’s seminal work~\cite{knuth1998art}, many variants of hash tables have been proposed. Among those directly related to this article, one can mention the paper by Knuth and Amble~\cite{AmbleKnuth74}, which introduces and studies a notion related to adding an overflow bit. More recently, the work of Larson~\cite{larson1983analysis} provides an analysis of hash tables with buckets and using 0 or 1 overflow bits, and Martini and Burkhard~\cite{martini2003double} started to study what happens with more overflow bits, but with algorithms that differ slightly from those that were developed many years later, when SIMD instructions became fully accessible: it is not exactly the same model we use this article, as, for them, $h_o(x)$ remains the same throughout the sequence of probes.

\begin{figure}[t]
\noindent
\begin{tikzpicture}[scale=.6]
\small
\begin{axis}[
    width=12cm,
    height=8cm,
    xlabel={Proportion of table occupancy},
    ylabel={Proportion of full groups},
    legend style={at={(0.03,0.97)},anchor=north west},
    grid=both,
    xticklabel={
        \pgfmathparse{\tick*100}
        \pgfmathprintnumber{\pgfmathresult}\%
    },
    yticklabel={
        \pgfmathparse{\tick*100}
        \pgfmathprintnumber{\pgfmathresult}\%
    },
    ylabel style={yshift=-1em},
    xticklabel style={font=\footnotesize},
    yticklabel style={font=\footnotesize}
]

\addplot[
    red,
    very thick,
]
coordinates {
(0.4375, 0.0032332703162837686) 
(0.44625, 0.0038664349793662263) 
(0.45499999999999996, 0.004597925009071537) 
(0.46374999999999994, 0.005438704349688519) 
(0.4724999999999999, 0.006400386038227865) 
(0.4812499999999999, 0.0074952135790481355) 
(0.4899999999999999, 0.008736038067351121) 
(0.49874999999999986, 0.010136291513201055) 
(0.5074999999999998, 0.011709956876291751) 
(0.5162499999999999, 0.01347153536650904) 
(0.5249999999999998, 0.015436011596209482) 
(0.5337499999999997, 0.01761881718736575) 
(0.5424999999999998, 0.020035793441071065) 
(0.5512499999999998, 0.02270315366949181) 
(0.5599999999999997, 0.025637445772650303) 
(0.5687499999999996, 0.02885551561607318) 
(0.5774999999999997, 0.032374471732153576) 
(0.5862499999999997, 0.03621165182991461) 
(0.5949999999999996, 0.04038459155658039) 
(0.6037499999999996, 0.04491099591175914) 
(0.6124999999999996, 0.04980871367281772) 
(0.6212499999999996, 0.055095715149739335) 
(0.6299999999999996, 0.06079007355083569) 
(0.6387499999999995, 0.06690995020839773) 
(0.6474999999999995, 0.07347358388686956) 
(0.6562499999999996, 0.08049928437645211) 
(0.6649999999999995, 0.08800543056314336) 
(0.6737499999999994, 0.09601047316304434) 
(0.6824999999999994, 0.10453294231525842) 
(0.6912499999999995, 0.1135914602449536) 
(0.6999999999999994, 0.12320475923737766) 
(0.7087499999999993, 0.13339170520633575) 
(0.7174999999999994, 0.14417132719878398) 
(0.7262499999999994, 0.15556285325325836) 
(0.7349999999999993, 0.16758575312708957) 
(0.7437499999999992, 0.18025978853005642) 
(0.7524999999999993, 0.19360507165597995) 
(0.7612499999999993, 0.20764213299633844) 
(0.7699999999999992, 0.22239199966141085) 
(0.7787499999999992, 0.23787628573837769) 
(0.7874999999999992, 0.25411729660071203) 
(0.7962499999999992, 0.27113814957439886) 
(0.8049999999999992, 0.2889629139989451) 
(0.8137499999999991, 0.30761677454360825) 
(0.8224999999999991, 0.3271262227207775) 
(0.8312499999999992, 0.3475192829779384) 
(0.8399999999999991, 0.36882578169114055) 
(0.848749999999999, 0.3910776700395669) 
(0.857499999999999, 0.4143094154341466) 
(0.8662499999999991, 0.43855848139798714) 
(0.874999999999999, 0.46386592333184345) 
(0.8837499999999989, 0.4902771387022708) 
(0.892499999999999, 0.5178428269545531)
};
\addlegendentry{Predicted}

\addplot[
    blue,
    mark=*,
    mark size=1.5pt,
     dashed,
]
coordinates {
(0.4375, 0.00341796875) 
(0.4462483723958333, 0.00396728515625) 
(0.4549967447916667, 0.0047607421875) 
(0.46374918619791666, 0.0054931640625) 
(0.47249755859375, 0.00640869140625) 
(0.48124593098958335, 0.00750732421875) 
(0.4899983723958333, 0.008544921875) 
(0.49874674479166664, 0.01043701171875) 
(0.5074991861979167, 0.01226806640625) 
(0.51624755859375, 0.0140380859375) 
(0.5249959309895833, 0.01611328125) 
(0.5337483723958333, 0.01849365234375) 
(0.5424967447916667, 0.02178955078125) 
(0.5512491861979166, 0.0244140625) 
(0.55999755859375, 0.02789306640625) 
(0.5687459309895834, 0.0311279296875) 
(0.5774983723958333, 0.03472900390625) 
(0.5862467447916667, 0.03839111328125) 
(0.5949991861979167, 0.04248046875) 
(0.60374755859375, 0.04632568359375) 
(0.6124959309895833, 0.0504150390625) 
(0.6212483723958333, 0.0552978515625) 
(0.6299967447916667, 0.0601806640625) 
(0.6387491861979167, 0.06597900390625) 
(0.64749755859375, 0.07196044921875) 
(0.6562459309895833, 0.0792236328125) 
(0.6649983723958334, 0.0867919921875) 
(0.6737467447916666, 0.0938720703125) 
(0.6824991861979167, 0.10333251953125) 
(0.69124755859375, 0.11236572265625) 
(0.6999959309895833, 0.1221923828125) 
(0.7087483723958333, 0.13214111328125) 
(0.7174967447916667, 0.1436767578125) 
(0.7262491861979167, 0.1549072265625) 
(0.73499755859375, 0.1671142578125) 
(0.7437459309895833, 0.18084716796875) 
(0.7524983723958333, 0.193359375) 
(0.7612467447916667, 0.20758056640625) 
(0.7699991861979166, 0.22357177734375) 
(0.77874755859375, 0.23809814453125) 
(0.7874959309895834, 0.25518798828125) 
(0.7962483723958333, 0.27154541015625) 
(0.8049967447916667, 0.28802490234375) 
(0.8137491861979167, 0.3070068359375) 
(0.82249755859375, 0.32427978515625) 
(0.8312459309895833, 0.34332275390625) 
(0.8399983723958333, 0.364013671875) 
(0.8487467447916667, 0.38446044921875) 
(0.8574991861979167, 0.40869140625) 
(0.86624755859375, 0.43084716796875) 
(0.8749959309895833, 0.457275390625) 
(0.8837483723958334, 0.48370361328125) 
(0.8924967447916666, 0.51025390625)
};
\addlegendentry{Experiment}

\end{axis}
\end{tikzpicture}
\hfill
\begin{tikzpicture}[scale=.6]
\small
\begin{axis}[
    width=12cm,
    height=8cm,
    xlabel={Proportion of table occupancy},
    ylabel={Extra hops for successful search},
    legend style={at={(0.03,0.97)},anchor=north west},
    grid=both,
    xticklabel={
        \pgfmathparse{\tick*100}
        \pgfmathprintnumber{\pgfmathresult}\%
    },
    ylabel style={yshift=-1em},
    yticklabel style={font=\footnotesize,/pgf/number format/fixed,/pgf/number format/precision=3},
    xticklabel style={font=\footnotesize},
    scaled ticks=false
]

\addplot[
    red,
    very thick,
]
coordinates {
(0.4375, 0.000306952524666837)
(0.44625, 0.0003706343007365831)
(0.455, 0.00044507399725168817)
(0.46375, 0.0005316576320255617)
(0.4725, 0.0006318893652592639)
(0.48125, 0.0007473954955561357)
(0.49, 0.0008799283143050562)
(0.49875, 0.0010313698697193008)
(0.5075, 0.0012037357040137753)
(0.51625, 0.0013991786391209793)
(0.525, 0.0016199926979816315)
(0.53375, 0.0018686172597872641)
(0.5425, 0.0021476415587573516)
(0.55125, 0.002459809647207045)
(0.56, 0.002808025955099769)
(0.56875, 0.0031953615902662325)
(0.5775, 0.0036250615364086247)
(0.58625, 0.004100552920387246)
(0.595, 0.004625454536654194)
(0.60375, 0.005203587835733492)
(0.6125, 0.005838989606147926)
(0.62125, 0.006535926606077158)
(0.63, 0.007298912433433058)
(0.63875, 0.008132726962274628)
(0.6475, 0.009042438721193147)
(0.65625, 0.010033430647467867)
(0.665, 0.011111429721869959)
(0.67375, 0.012282541076072961)
(0.6825, 0.013553287271576675)
(0.69125, 0.01493065358080492)
(0.7, 0.01642214026391553)
(0.70875, 0.01803582303702922)
(0.7175, 0.019780423179681872)
(0.72625, 0.021665389045280037)
(0.735, 0.023700991136660354)
(0.74375, 0.025898433414148547)
(0.7525, 0.028269984148918903)
(0.76125, 0.030829130464989918)
(0.77, 0.03359076179059972)
(0.77875, 0.03657138884951894)
(0.7875, 0.039789406685025286)
(0.79625, 0.04326541269382078)
(0.805, 0.04702259399844233)
(0.81375, 0.051087203060594666)
(0.8225, 0.05548914676054939)
(0.83125, 0.060262723030140686)
(0.84, 0.065447551738683)
(0.84875, 0.07108976477886175)
(0.8575, 0.07724354718839455)
(0.86625, 0.0839731615730217)
(0.875, 0.09135565026664727)
(0.88375, 0.09948450769820805)
(0.8925, 0.10847477445981046)
};
\addlegendentry{Predicted}

\addplot[
    blue,
    mark=*,
    mark size=1.5pt,
     dashed,
]
coordinates {
(0.4375, 0.0003255208333333333)
(0.4462483723958333, 0.00039208534694994075)
(0.4549967447916667, 0.0004650330888928635)
(0.46374918619791666, 0.0005966430056768827)
(0.47249755859375, 0.0007147716605954134)
(0.48124593098958335, 0.0007525090681570292)
(0.4899983723958333, 0.0008719336998222916)
(0.49874674479166664, 0.001019808765460301)
(0.5074991861979167, 0.0011946473385021207)
(0.51624755859375, 0.0014502691668045998)
(0.5249959309895833, 0.0016353673376064734)
(0.5337483723958333, 0.001837254333938128)
(0.5424967447916667, 0.0020926464852539677)
(0.5512491861979166, 0.002502306698652888)
(0.55999755859375, 0.0028555858310626704)
(0.5687459309895834, 0.0031979967805401536)
(0.5774983723958333, 0.003579330073418542)
(0.5862467447916667, 0.004115883283822427)
(0.5949991861979167, 0.004718690802656144)
(0.60374755859375, 0.005324275325690639)
(0.6124959309895833, 0.005872700578633734)
(0.6212483723958333, 0.006392538545173502)
(0.6299967447916667, 0.007253210013692613)
(0.6387491861979167, 0.008262251638754228)
(0.64749755859375, 0.00903040929057557)
(0.6562459309895833, 0.010131511232088492)
(0.6649983723958334, 0.011252524016398459)
(0.6737467447916666, 0.012368643555985022)
(0.6824991861979167, 0.013748203969451086)
(0.69124755859375, 0.015034053249039033)
(0.6999959309895833, 0.01659584609750568)
(0.7087483723958333, 0.018193613576603783)
(0.7174967447916667, 0.02002472608488533)
(0.7262491861979167, 0.021738765036446048)
(0.73499755859375, 0.023711060548183335)
(0.7437459309895833, 0.025954273646892765)
(0.7524983723958333, 0.02845880151837953)
(0.7612467447916667, 0.030857796497829854)
(0.7699991861979166, 0.033619573546119906)
(0.77874755859375, 0.036507563288658985)
(0.7874959309895834, 0.03991009378148655)
(0.7962483723958333, 0.0435186983228233)
(0.8049967447916667, 0.04711478194059726)
(0.8137491861979167, 0.05127833309165096)
(0.82249755859375, 0.05581857848884667)
(0.8312459309895833, 0.06048353541830856)
(0.8399983723958333, 0.06554994719964348)
(0.8487467447916667, 0.0711977678485819)
(0.8574991861979167, 0.07734211512819174)
(0.86624755859375, 0.0840484947554829)
(0.8749959309895833, 0.09142992666446552)
(0.8837483723958334, 0.09909296008103503)
(0.8924967447916666, 0.10779611561958603)
};
\addlegendentry{Experiment}

\end{axis}

\end{tikzpicture}
\caption{On the left, the proportion of full buckets (containing $b$ elements) as the occupancy increases. On the right, the expected number of additional hops for a successful search. The blue dots were obtained experimentally, using $b=16$, $n=2^{14}=16\,384$ and only one trial to illustrate the fact that our statements hold with high probability. The red line corresponds to the theoretical solutions  established in this paper, Theorem~\ref{thm:proportion-buckets-no-overflow} and Corollary~\ref{cor:successful-search}. Though asymptotic, our results fit the experiments particularly well. \label{fig:experiments}}
\end{figure}
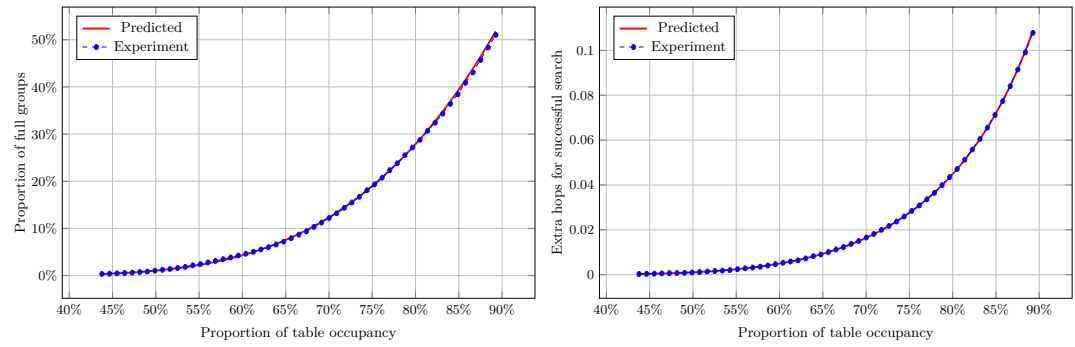

Although the papers~\cite{larson1983analysis,martini2003double} correctly state some weaker versions of some results presented in Section~\ref{sec:no overflow}, they are not entirely complete, in the sense that some proofs are more heuristic than fully rigorous. For example, certain arguments based on expectations would require concentration around the mean to be made rigorous, and approximating the number of times a slot is accessed by independent Poisson distributions is only asymptotically valid, requiring careful handling of error terms. 
The equations that appear when studying the dynamics of that kind of hash tables are indeed quite involved, and probably very difficult to handle using elementary combinatorial and probabilistic techniques. More advanced methodologies, as concentration inequalities and Poissonization-Depoissonization techniques, should help to finalize the proofs, but it is not the path we followed  in the sequel.

Indeed, we use results stemming from the work of Wormald~\cite{wormald95,Wormald1999TheDE} on dynamic graphs  which were later revisited by Warnke~\cite{warnke2019wormaldsdifferentialequationmethod}, tailored to our needs, and which are particularly well suited to studying the dynamics of the parameters of interest in our hash tables. 
These theorems are powerful and make it possible to obtain results that hold with high probability during the whole trajectory, which is far more useful than results in expectation, and they also provide usable bounds on the error terms. In Section~\ref{sec:no overflow}, we present this method and explain how to use it to precisely analyze hash tables without overflow bits. More precisely, we can analyze after $k$ insertions the number of buckets containing $i$ elements, simultaneously for every $i\in\{0,\ldots,b\}$. Theorem~\ref{thm:proportion-buckets-no-overflow} states that with high probability, this number is $n\,u_i(\frac{k}{nb})+o(n)$, with $u_i(t)=\frac1{i!}\lambda_b(t)^ie^{-\lambda_b(t)}$, where $\lambda_b(t)$ is  implicitly defined by a single equation. In Theorem~\ref{lambda-number-hops} we link  $\lambda_b(t)$ to the typical number of hops, by proving that with high probability it is $n\lambda_b(\frac{k}{nb})+o(n)$. Consequently, we obtain the typical cost of a successful and unsuccessful search in a table containing $k$ elements, in Corollary~\ref{cor:successful-search}
 and Corollary~\ref{cor:unsuccessful-search}. Our theoretical results allow a very precise description of the behavior of SIMD hash tables, as depicted in Fig.~\ref{fig:experiments}. Since our statements hold with high probability, which is much stronger than in expectation, only one trial was used for the experimental plots. 

To conclude this related work section, we would like to mention two recent breakthroughs in the theoretical study of hash tables, although the problems they address are orthogonal to those studied in this paper, which focuses on how hash tables are implemented in contemporary programming languages. The first is Kuszmaul and Xi's paper~\cite{kuszmaul2024}, which makes the first tangible progress on a question open since 1968 by proving that quadratic probing supports constant-time operations at small constant load factors, a result that extends more generally to any fixed-offset open-addressing scheme. The second is due to Bender, Kuszmaul and Zhou~\cite{bender2024}. They show that, if insertions and deletions are allowed to rearrange existing elements (see also~\cite{FarachColtonKK24,FCKK26} for recent results when reordering the keys is not allowed), classical open addressing can achieve 
$\O(1)$ expected query time with only $\O(\log \log \varepsilon^{-1})$ expected insertion/deletion time at load factor $1-\varepsilon$, via a new ``rainbow hashing'' scheme. 

placed. 
They  also established a matching $\Omega(\log \log \varepsilon^{-1})$ lower bound for the amortized expect time of insertions and deletions, when the successful search time is $\O(\log \log \varepsilon^{-1})$, hence proving the optimality of their solution.

\paragraph*{Plan of the article and contributions}
In Section~\ref{sec:no overflow} we show how to apply Wormald's Differential Equation method to study the tables without overflow. In particular we obtain precise estimates of the total number of buckets with a given number of occupied slots in Theorem~\ref{thm:proportion-buckets-no-overflow} and the total number of hops in Theorem~\ref{lambda-number-hops}. These results involve an implicitly defined function $\lambda_b(t)$ tightly related to the number of hops.
After presenting some properties of $\lambda_b(t)$ in Section~\ref{sec:lambda}, we move on to study the hash tables with overflow bits in Section~\ref{sec:analysis-overflow}. To overcome technical difficulties, we change the scale of (discrete) time by considering every insertion attempt, successful or not, as one time unit. We can then derive the system of equations for the dynamics of the table, considering simultaneously the number of buckets with $i$ elements for $i\in\{0,\ldots,b-1\}$ and the number of full buckets with $j$ overflow bits set to 1, for $j\in\{0,\ldots,d\}$. As Wormald's technique provides tight asymptotic results with high probability, we obtain a precise estimation of these quantities, in Theorem~\ref{thm:proportion-buckets-overflow}, and of the cost of an unsuccessful search, in Theorem~\ref{thm:unsuccessful-overflow} (insertions and successful searches are not affected by overflow bits). 

\section{Analysis without overflow}\label{sec:no overflow}

Throughout the article, the hash table is made of $n$ buckets, each containing up to $b$ elements. We denote by $U_j(k)$ the number of buckets with exactly $j$ elements at time $k$. Notice that $U_0(0)=n$ and $U_j(0)=0$ for $j\in[b]$, where $[b]:=\{1,\ldots,b\}$, since we start with an empty table (or just after resizing, see the discussion in Section~\ref{sec:intro}). Using the uniform hashing with random probing model, each of $U_j(k)$ 
is a random variable, which is our main interest for the analysis. They are linked by the equation $\sum_{j=0}^b U_j(k) = n$, which always holds.

\subsection{Recurrence relations}

In our model, the probe sequence when looking for a bucket for inserting a new element at time $k+1$ is simply a sequence of independent uniform random variables on $\{0,\ldots,n-1\}$. The bucket it is inserted into goes from $j$ to $j+1$ elements, so $U_j(k+1) = U_j(k)-1$ and $U_{j+1}(k+1) = U_{j+1}(k)+1$. Since it cannot be inserted in one of the $U_b(k)$ buckets that are already full, by uniformity this yields:
\begin{equation}
    U_j (k+1) - U_j(k) = \begin{cases}
         +1, & \text{with probability } \frac{U_{j-1}(k)}{n-U_b(k)} \text{ if } j>0\,, \\
         -1, & \text{with probability } \frac{U_{j}(k)}{n-U_b(k)} \text{ if } j<b\,,\\ 
         0, & \text{otherwise}.
    \end{cases}
\end{equation}
If $\mathcal{F}_k$ denotes the $\sigma$-algebra of all random choices up to time $k$, this means that the conditional expectation of the difference, which is a random variable, satisfies
\begin{equation}
    \label{eq:conditional-expectation-bucketsu}
    \mathbb{E}[U_j (k+1) - U_j(k) \,\,\big|\,\, \mathcal{F}_k] = {\mathbf{1}}_{j>0}\frac{\tfrac{U_{j-1}(k)}{n}}{1-\frac{U_b(k)}{n}} - {\mathbf{1}}_{j<b} \frac{\tfrac{U_{j}(k)}{n}}{1-\frac{U_b(k)}{n}}\,,
\end{equation}
where ${\mathbf{1}}_{P}$ is $1$ if $P$ holds and $0$ otherwise, is used to account for the border effects.

That kind of system of equations is amenable to our main tool, Wormald's Theorem, which gives precise asymptotic estimates for the $U_j(k)$'s, with high probability. In a nutshell, it consists in moving to continuous time, using the \emph{fluid limits} $u_j(t) =\lim_n U_j(\lfloor t b n\rfloor)/n$, for $t\in [0,1)$. Eq.~\eqref{eq:conditional-expectation-bucketsu} directly produces a system of differential equations $u_j^\prime(t)= f_j(u_0(t),\ldots,u_b(t))$, where\footnote{The factor $\bucketSize$ appearing in $f_j$ but not in Eq.~\eqref{eq:conditional-expectation-bucketsu} comes from the fact that the time normalization is performed with respect to $n\cdot \bucketSize$ rather than $n$.}
$f_j(u_0,\ldots,u_b):=b\cdot{\mathbf{1}}_{j>0}\frac{u_{j-1}}{1-u_b}-b\cdot{\mathbf{1}}_{j<b} \frac{u_j}{1-u_b}$. A variant of Wormald Theorem (see Section~\ref{sec:wormaldstheorem} for the details) then implies that
 with high probability $U_j(\lfloor t bn\rfloor) = u_j(t) n +o(n)$. After this, it remains to study the solutions $u_j$ of the system of differential equations.

\subsection{Wormald's differential equation method}
\label{sec:wormaldstheorem}

Fix a bounded, connected and open set $\mathcal{D}\subseteq \mathbb{R}\times \mathbb{R}^\ell$. Consider functions $f_j\colon\mathcal{D}\to \mathbb{R}$ that  are continuous and satisfy a Lipschitz condition: there exists $L>0$ such that, for all $(t_1;u_1,\ldots,u_\ell),(t_2;w_1,\ldots,w_\ell)\in U$, and for all $j\in\{1,\ldots,\ell\}$  \begin{equation*}
   |f_j(t_1;u_1,\ldots,u_\ell)-f_j(t_2;w_1,\ldots,w_\ell)| \leq L\times \max \{|t_1-t_2|,|u_1-w_1|,\ldots,|u_\ell-w_\ell|\}. 
\end{equation*} 
Suppose $(t_0,v_1,\ldots,v_\ell)\in \mathcal{D}$, then the well-known Picard-Lindel\"of Theorem (see for instance~\cite[Theorem 7.19]{apostol}) ensures that, for $t$ in some open interval containing $t_0$, there is a unique solution $(y_1(t),\ldots,y_b(t))$ to the system of ordinary differential equations:
\begin{equation}
    \label{eq:sys-ode}
    y_i'(t) =  f_i(t;y_1(t),\ldots,y_\ell(t)),\ \forall i\in\{1\,\ldots,\ell\},
\end{equation}
with the initial condition $y_1(t_0)=v_1,\ldots,y_\ell(t_0)=v_\ell$.  
In principle, this statement is local, but, by connectedness, the solution can be extended as long as $(t,\boldsymbol{y}(t))$ is in $\mathcal{D}$. The solution stops existing, if ever, when $(t,\boldsymbol{y}(t))$ reaches the boundary of the domain $\mathcal{D}$. 

We can now present our main tool: Wormald's differential equation method. The statement of Theorem~\ref{thm:wormald-probabilistic-increment} is an adaptation of Warnke's result~\cite{warnke2019wormaldsdifferentialequationmethod}, which revisits and makes more precise  the original theorems of Wormald~\cite{wormald95,Wormald1999TheDE}. We describe the different parameters intuitively after the statement of Theorem~\ref{thm:wormald-probabilistic-increment}.

\begin{theorem}[Wormald's method for increments bounded with high probability]

    \label{thm:wormald-probabilistic-increment}
        Let $\ell\geq 1$ be a fixed integer. Consider discrete random processes $(Y_j)_{j=1}^\ell$ taking real values, where $Y_j(k)$ is the value of the $j$-th process at time $k$. 

    Consider an integer variable
    $N\geq 1$. Suppose that there are functions $m:=m(N)$, $\beta=\beta(N)$, $\gamma:=\gamma(N)$ and $\delta:=\delta(N)$, such that $\lim_N\gamma(N) = \lim_N\delta(N)=0$ and 
    \begin{enumerate}
        \item  for all $k<m$ and $j\in [\ell]$,
        $
        \Pr(|Y_j(k+1)-Y_j(k)|\leq \beta(N))\geq 1 - \gamma(N)\,;
        $
        \item for all $j\in [\ell]$, and  $k<m,$
        \[
        \mathbb{E}[Y_j(k+1)-Y_j(k) \,\,\big|\,\, \mathcal{F}_k] = f_j(k/N; Y_1(k)/N,\ldots,Y_\ell(k)/N)\,,
        \]
        where the functions $f_j $ are continuous, uniformly bounded by $R>0$, and are $L$-Lipschitz on an open and connected domain $\mathcal{D}$, for some $L>0$. 
        \item The initial vector $\boldsymbol{Y}(0)=(Y_1(0),\ldots,Y_\ell(0))$ satisfies $\|{\boldsymbol{Y}}(0)-N \boldsymbol{v}\|_\infty \leq \delta\cdot N$ for some fixed ${\boldsymbol{v}}=(v_1,\ldots,v_\ell)$ with $(0,v_1,\ldots,v_\ell)\in \mathcal{D}$.
    \end{enumerate}    
        Let  $\boldsymbol{y}(t)=(y_1(t),\ldots,y_\ell(t))$ be the unique solution to the system of differential equations of Eq.~\eqref{eq:sys-ode}, with initial condition $y_1(0)=v_1,\ldots,y_\ell(0)=v_\ell$. Let $\sigma>0$ be the supremum of $t\geq 0$ such that the solution $(s,\boldsymbol{y}(s))\in \mathcal{D}$ for all $s\in [0,t)$ and fix $\theta \in (0,\sigma)$.

        Whenever $\delta\geq R/N$, with probability at least        \reviewer{See Appendix~\ref{sec:proof-wormald-probabilistic-increment} for details.} $1-\theta \ell N \gamma-2 b \exp(-N \delta^2 / (8\theta \beta^2)) $, we have $|Y_j(k) - N y_j(k/N) |\leq 3 e^{L \theta} \delta \cdot N$, for all $j \in [\ell]$ and $k\leq \min(m,\theta N)$.

\end{theorem}

Let us explain the main parameters introduced in  Theorem~\ref{thm:wormald-probabilistic-increment}. Here $N$ controls the discrete time. Indeed, notice that the final implication is that $Y_j(k)\approx Ny_j(k/N)$ with high probability, where the $y_i$'s are the solutions of the system of differential equations defined by the $f_j$'s.
Hence, by renormalizing by $N$, we obtain in the limit a system of continuous equations, which provide an accurate first-order approximation of the $Y_j$.
For this result to be useful, we need that  $\gamma N$ tend to $0$ and  $\frac{N\delta^2}{\beta^2}$ tend to $\infty$. 
Hence, $\beta$ is relatively small and Condition~\textcolor{darkgray}{\bf 1} means that it bounds the increment of $Y_j$ with high probability. 
In the sequel, the value $\theta \in (0,1)$ corresponds to the \emph{maximum load factor} of the table (typically around 80\

The role of $m=m(N)$ is to delimit the applicability of the recurrence conditions~\textcolor{darkgray}{\bf 1} and~\textcolor{darkgray}{\bf 2}. In our case simply considering $m(N)=\lfloor\theta N\rfloor$ works. We have chosen to keep the parameters $m$ and $\sigma$ in the statement of the theorem to give a more general statement, since this method should apply for the analysis of the dynamic of many discrete systems.

The equation for the conditional expectation is given exactly in Condition~\textcolor{darkgray}{\bf 2} (called the \emph{trending hypothesis}) since it is sufficient for our needs, but can be weakened~\cite{wormald95,Wormald1999TheDE} to an approximation up to uniform $o(1)$.

The upper bound  of $\theta \ell N \gamma+2 \ell \exp(-N \delta^2 / (8\theta \beta^2)) $ on the probability of failure is obtained quite directly from Theorem~2 and Lemma~9 of Warnke’s article~\cite{warnke2019wormaldsdifferentialequationmethod}.\reviewer{See Appendix~\ref{sec:proof-wormald-probabilistic-increment} for details.}

\subsection{Results without overflow bits}
\label{sec:no-overflow}

Let $N=bn$ be the maximum number of elements in the hash table.
In what follows, we fix $\theta \in (0,1)$ and suppose that the time $k\leq \theta N$. This parameter $\theta$ can be seen as the classical \emph{maximal load factor}, the threshold that triggers resizing. Observe that, since there have been $k$ insertions we have $\sum_{i=0}^b U_i(k) i = k$, so that $U_b(k)\leq k/b$ and $U_b(k)/n \leq k / N \leq \theta$.

Applied to Eq.~\eqref{eq:conditional-expectation-bucketsu}, Wormald's method yields the following system for the associated $u_i(t)=\lim_n U(\lfloor tbn \rfloor)/n$. As the continuous time for $t$ is scaled to reaching a full table with $N=bn$ elements at time $t=1$, differentiating $tb$ produces an extra $b$ factor: 
\begin{equation}
    \label{eq:sys-ode-occupied-buckets-no-overflow}
    \begin{cases}
            u_0'(t) &=  - b\frac{u_0(t)}{1-u_b(t)}\,, \\
            u_i'(t) &=  b\frac{u_{i-1}(t)}{1-u_b(t)} - b\frac{u_i(t)}{1-u_b(t)}\,,\text{ for all }i\in\{1,\ldots, b-1\}, \\
     u_b'(t)&= b\frac{u_{b-1}(t)}{1-u_b(t)}\,.
    \end{cases}
\end{equation}    
We also have the following properties of this system:
\begin{itemize}
    \item We fix the initial condition $u_0(0)=1,u_1(0)=\ldots u_\bucketSize(0)=0$, which corresponds exactly to the initial condition for $(U_i)_{i=0}^b$, namely $U_0(0)=n,U_1(0)=\ldots=U_b(0)=0$. 

    \item The sum $\sum_{i=0}^b u_i(t) = 1$, since the  derivatives in Eq.~\eqref{eq:sys-ode-occupied-buckets-no-overflow} sum to zero.

    \item The sum $\sum_{i=0}^b u_i(t) i = b t$, which can be derived by simple computations.
\end{itemize}
  One can prove that\reviewer{See Lemma~\ref{lemma:sols-are-positive} in Appendix~\ref{sec:proof:thm:proportion-buckets-no-overflow} for a proof of $u_i(t)\geq 0$.} $u_i(t)\geq 0$, and that the latter property implies $u_b(t)\leq t$. Therefore, the system satisfies all of the hypothesis of the Picard-Lindel\"of Theorem and is well-defined defined in $t\in [0,1-\varepsilon)$, for every $\varepsilon$.

  Thus taking any $\theta<1$ is a valid choice, which allows us to  apply Theorem~\ref{thm:wormald-probabilistic-increment} to obtain that $U_i(k) = n u_i(k/N)+o(n)$.  The associated system of differential equations already appears, in a very similar form, in \cite[p. 13--18]{Wormald1999TheDE}, where it comes up as the approximation to a process (the $d$-process) of adding random edges to a graph, in such a way that the degree never exceeds $d$. Although the discrete objects counted are not the same, the system of differential equations is, but restricted only to $b=2$ in \cite{Wormald1999TheDE} due to the edges of the graphs. A direct adaptation gives the following solution of the system, characterized in terms of an implicit function $\lambda_b(t)$. \reviewer{See Apppendix~\ref{sec:proof:thm:proportion-buckets-no-overflow} for the remaining details of the proof.}

\begin{restatable}{theorem}{thmproportionoccupied}

\label{thm:proportion-buckets-no-overflow}
Fix $\theta \in (0,1)$. Starting from an empty table with $n$ buckets of size $\bucketSize$, the number $U_i(k)$ of buckets containing exactly $i \in [b]$ elements after $k\leq \theta N$ insertions satisfies
\[
U_i(k) = n  u_i(k/N) + o(n)\,,
\]
with high probability, where the $o$-term is uniform in $i$ and $k$ (but depends on $\theta$), and $u_1(t),\ldots,u_\bucketSize(t)$ are the solutions to the system of ordinary differential equations of Eq.~\eqref{eq:sys-ode-occupied-buckets-no-overflow} with initial condition $u_0(0)=1,u_1(0)=0,\ldots,u_\bucketSize(0)=0$.

Moreover,  for each $i \in \{0,1,\ldots,b-1\}$, we have 
$u_i(t) =  \tfrac{1}{i!} \lambda_b(t)^i e^{-\lambda_b(t)}$,
 where $\lambda_b(t)$ is defined implicitly from the equation $e^{-\lambda_b(t)} \sum_{i<b} (b-i)\tfrac{1}{i!} \lambda_b(t)^i  = b - b t$. 

\end{restatable}

\begin{remark}
    The condition $e^{-\lambda_b(t)} \sum_{i<b} (b-i)\tfrac{1}{i!} \lambda_b(t)^i  = b - b t$ follows from the equations $\sum_i u_i(t)=1$ and $\sum_i i u_i (t) = bt$ by subtracting. Two equations are needed, as we have to define both $u_b(t)=1-u_0(t)-\ldots-u_{b-1}(t)$  and $\lambda_b(t)$.
\end{remark}

Differentiating $u_0(t) = e^{-\lambda_b(t)}$,  we obtain $u_0'(t) = - u_0(t) \lambda_b'(t)$. Using Eq.~\eqref{eq:sys-ode-occupied-buckets-no-overflow}, we have $\lambda_b'(t) = -\tfrac{u_0'(t)}{u_0(t)} = \tfrac{1}{1-u_b(t)}$. This proves the following lemma.

\begin{lemma}
    \label{lemma:lambda-integral}
    For each $t\in [0,1)$, $\lambda_b(t) = b \int_0^t \frac{dt}{1-u_b(t)}$.
\end{lemma}

Lemma~\ref{lemma:lambda-integral} leads to a connection between $\lambda_b$ and the total number of hops $R(k)$, stated formally in Theorem~\ref{lambda-number-hops}. Indeed, the expected increment of $R(k)$ is  $\mathbb{E}[ R(k+1)-R(k)  \big| \mathcal{F}_k] = \sum_{r\geq 1} r (1-U_b(k)/n) (U_b(k)/n)^{r-1}  = \frac{1}{1-\nicefrac{U_b(k)}{n}}$. 
Thus Wormald's Method almost applies here, the only detail to deal with being the probabilistic upper bound of $R(k+1)-R(k)$, Condition~\textcolor{darkgray}{\bf 1} of Theorem~\ref{thm:wormald-probabilistic-increment}. This is done using Lemma~\ref{lemma:increment-small-rk} below, which provides a bound to the increments of $R(k)$, with high probability. 
\begin{restatable}{lemma}{incrementsmallrk}

    \label{lemma:increment-small-rk}
    Fix $\theta \in (0,1)$. The total number of hops satisfies $R(k+1)-R(k) > (\log N)^2$ with probability at most $\tfrac{1}{\theta(1-\theta)}\exp(-(1-\theta) \times (\log N)^2)$, for each $k\leq \theta N$.

\end{restatable}
\begin{proof}
    Note that $U_b(k)\leq k/b$, as a bucket must have $b$ elements to be full. Thus, $U_b(k)/n \leq k/N\leq \theta$. 
    Given $\mathcal{F}_k$, the number of hops follows a geometric law of parameter $U_b(k)/n$, so \[\Pr( R(k+1)-R(k) = r \big| \mathcal{F}_k) = (1-U_b(k)/n) (U_b(k)/n)^{r-1} \leq \theta^{r-1}\,.\]
    Thus, by the tower property of conditional expectation, \[\Pr(R(k+1)-R(k) = r) = \mathbb{E}[\Pr( R(k+1)-R(k) = r \big| \mathcal{F}_k) ] \leq \theta^{r-1}\,. \]
    Hence, $\Pr(R(k+1)-R(k) >  x) \leq \tfrac{\theta^{\lfloor x\rfloor}}{1-\theta} \leq \frac{\theta^{x-1}}{1-\theta}$. The formula follows from $\log \theta \leq \theta - 1$.
    
\end{proof}

All the ingredients are now established to estimate the total number of hops, which is also the cost of the construction of the hash table, starting from an empty one.

\begin{restatable}{theorem}{lambdanumberhops}

    \label{lambda-number-hops}
    Fix $\theta \in (0,1)$, the total number of hops satisfies $R(k) = \tfrac{N}{b} \lambda_b(k/N) + o(n)$ with probability tending to one, for all $k\leq \theta N$, where the $o$-term is uniform in $k$.

\end{restatable}
\begin{proof}
We apply Theorem~\ref{thm:wormald-probabilistic-increment} to the trending hypothesis $\mathbb{E}[ R(k+1)-R(k)  \big| \mathcal{F}_k] = \frac{1}{1-{U_b(k)}/{n}}$, plus Eq.~\ref{eq:conditional-expectation-bucketsu} in order to complete the system with $b+2$ equations, $R(k),U_0(k),\ldots,U_b(k)$.

Lemma~\ref{lemma:increment-small-rk} tells us that $R(k+1)-R(k) \leq (\log N)^2$ is satisfied for any $k\leq \theta N$ with probability at least $1-  \tfrac{1}{1-\theta}\exp(-(1-\theta) \times (\log N)^2)$. Let $\beta(N):=(\log N)^2$ and $\gamma(N):=\tfrac{1}{1-\theta}\exp(-(1-\theta) \times (\log N)^2)$, so that hypothesis $(1)$ of Theorem~\ref{thm:wormald-probabilistic-increment} is satisfied.

The function $\delta=\delta(N)$ in Theorem~\ref{thm:wormald-probabilistic-increment}, in this case should satisfy $\delta\geq \tilde{R}/N$ and tend to $0$, where $\tilde{R}=\tfrac{1}{1-\theta}$, as $u_b(t)\leq \theta<1$. At this point we have to check that $\theta b N \gamma(N)\to 0$, and $N \delta^2 / (8 \theta \beta^2)\to\infty$ for an appropriate choice of $\delta=\delta(N)$. Clearly, $N \gamma(N) \to 0$. Picking $\delta (N) = N^{-1/2+\varepsilon}$, for fixed $\varepsilon\in(0,\tfrac{1}{2})$, makes $N \delta^2 / (8 \theta \beta^2) = N^{2\varepsilon} / (8 \theta (\log N)^4) \to\infty$. Thus the probability tends to one. Finally, the error term is $O(\delta(N) N) = O(N^{1/2+\varepsilon})=o(n)$.
\end{proof}

To analyze the expected cost of searches in an already built hash table, one usually adds a level of randomness for such elementary operations. The classical setting~\cite{knuth1998art} is to consider that all elements are equally likely to be queried during a successful search,  and that an independent uniform random sequence of probes is used for a unsuccessful search. In the statements below, the \emph{with high probabilities} refers to the construction of the hash table, and the \emph{average} cost to the random choice of the element that is searched. 
For the result below, we need to add the condition that $t\in(\vartheta,\theta)$ for some $\vartheta>0$ to have the uniformity of the error term $o(1)$. It is not a restriction, as the full hash table mechanism usually involves resizing, so it cannot be almost empty, except  at the very beginning of the building process. 

Corollary~\ref{cor:successful-search} follows from Theorem~\ref{lambda-number-hops}, as the expected number of hops in a successful search is $R(k)/k$. Indeed, the number of hops needed to insert an element corresponds to its look-up time, thus $R(k)$ is the sum of all of the look-up times, while the factor $1/k$ corresponds to the choice of the element to look-up as the table has $k$ elements, since it is chosen uniformly at random amongst the $k$ elements in the table.

\begin{corollary}[Successful search]
    \label{cor:successful-search}
    For any $\vartheta \in (0,\theta)$, with high probability, the expected number of hops in a successful search in a hash table with $k$ elements is $\lambda_b(t)/(b t) + o(1)$, where $t=k/N$ and the $o$ term is uniform  for $t\in(\vartheta,\theta)$.
\end{corollary}

 Corollary~\ref{cor:unsuccessful-search} is a direct consequence of Theorem~\ref{thm:proportion-buckets-no-overflow}, since the  expected number of hops for an unsuccessful search follows a geometric law of parameter $U_b(k)/n$. This can also be interpreted as the number of hops for inserting a new element, since we have no overflow bit in this section to stop the sequence of probes of a unsuccessful search prematurely.

\begin{corollary}[Unsuccessful search with no overflow bits]\label{cor:unsuccessful-search}
    With high probability, the expected number of hops in a unsuccessful search in a hash table with $k$ elements is $(1-u_b(t))^{-1} + o(1)$, where $t=k/N$ and the $o$ term is uniform for $t\in [0,\theta)$.
\end{corollary}

\section{Some properties of the number of hops function $\lambda_b(t)$}\label{sec:lambda}

Since the function $\lambda_b(t)$ turns out to be particularly important for the study of modern hash tables, we devote a short section to give some of its properties. First, for small $b$ one has:

\begin{itemize}
    \item If $b=1$, we simply have $\lambda_1(t)=\log (\frac{1}{1-t})$. Then the successful search time from Corollary~\ref{cor:successful-search} coincides with the classical $\frac{1}{t}\log (\frac{1}{1-t})$ (see~\cite{knuth1998art}).

    \item If $b=2$, then $2-2t = (2+\lambda_b(t)) e^{-\lambda_b(t)}$. This equation can be solved in terms of the  $-1$-branch $W_{-1}(x)$  of the Lambert-W function~\cite{corless1996lambert} as $\lambda_b(t)=-2-W_{-1}\left(-\tfrac{2-2t}{e^2}\right)$. 
\end{itemize}

Our first result states that $\lambda_b(t)$ is always of logarithmic magnitude. \reviewer{See Appendix~\ref{sec:proofpropertieslambda} for the proofs of this section.}

\begin{restatable}{proposition}{propertieslambda}

    \label{properties-lambda}
    The inequalities $\log\left(\tfrac{1}{1-t}\right) \leq \lambda_b(t)\leq b \log\left(\tfrac{1}{1-t}\right)$ are satisfied for all $t$. Moreover  $\lambda_b(t) \sim \log \left(\tfrac{1}{1-t}\right) $ as $t\to 1^-$ and $\lambda_b(t) \sim  b \log \left(\tfrac{1}{1-t}\right) $ as $t\to 0^+$.

\end{restatable}

Combined with Corollary~\ref{cor:successful-search}, Proposition~\ref{properties-lambda} shows shows that the classical framework without SIMD instructions uses $b$ times as many instructions when the table (of same capacity) is almost full, and roughly the same number when it is almost empty. The latter is to be expected, as the element is usually found with only one probe in this case. It is not an issue in practice, since after the first resizing mechanism is triggered, the hash table always has a load factor greater than some positive $\vartheta$; for example, if $\vartheta=1/2$, a SIMD-based hash table uses, in the limit, at least 27\

The following proposition provides an alternative characterization of $\lambda_b(t)$ by means of an ordinary differential equation involving the classical upper incomplete $\Gamma$-function $\Gamma(a,x)$ and the lower incomplete $\Gamma$-function  $\gamma(a,x)$, for $a>0$ and $x\geq0$,
\[
\Gamma(a,x):=\int_x^\infty t^{a-1} e^{-t} dt\,,\qquad \gamma(a,x):=\int_0^x t^{a-1} e^{-t} dt\,,\qquad \Gamma(a):=\int_0^\infty t^{a-1} e^{-t} dt\,.
\]
In particular, their regularized versions $Q(a,x)={\Gamma(a,x)}/{\Gamma(a)}$, the regularized upper-incomplete $\Gamma$-function, and  the regularized lower-incomplete $\Gamma$-function $P(a,x) = {\gamma(a,x)}/{\Gamma(a)}$. 

The net gain here is that these functions have efficient implementations~\cite{Temme_1994}, that are implemented in mathematical packages such as {\tt SciPy}, {\tt mpmath}\footnote{For instance, with {\tt mpmath} in {\tt Python3}, the function $\lambda_b(t)$ can be implemented as a one-liner: \\ {\tt mpmath.odefun(lambda t, y : b/mpmath.gammainc(b,y,mpmath.inf,regularized=True),0,0)}.} or {\tt boost}, amongst many.
\begin{restatable}{proposition}{gammaincomplete}

    The function $\lambda_b(t)$ satisfies $\lambda_b^\prime(t)= \nicefrac{b}{Q(b,\lambda_b(t))}$, while  
     $u_b(t)=P(b,\lambda_b(t))$.

\end{restatable}

\section{Analysis with overflow}
\label{sec:analysis-overflow}

In this section, we analyze the impact of having $d$ overflow bits, by proving Theorem~\ref{thm:proportion-buckets-overflow}, where $V_j(k)$ is the number of full buckets having $j$ bits at 1 at time $k$, for $j\in\{0,\ldots,d\}$. The use of integral forms may seem complicated at first, but they allow for very fast and accurate numerical estimations, this is what we used to obtain the curves given in this paper. 

\begin{restatable}{theorem}{thmpropbucketsoverflow}
    
    \label{thm:proportion-buckets-overflow}
Fix $\theta \in (0,1)$. Starting from an empty table with $n$ buckets of size $b$ and $d$ overflow bits, the number $V_j(k)$ of full buckets with $j$ overflow bits on at time $k$ satisfies $V_j(k) = n v_j(k/N)+o(n)$ with high probability and uniformly for all $k\leq \theta N$, 
where 
\begin{equation}
    \label{eq:system-overflow-bits}
    \begin{cases}
            {v}_0'(t) &=  b \frac{{u}_{b-1}(t)}{1-u_b(t)} - b \frac{{v}_0(t)}{1-u_b(t)}\,, \\
            {v}_1'(t) &= b(1-\tfrac{0}{d})\frac{{v}_{0}(t)}{1-u_b(t)} - b(1-\tfrac{1}{d}) \frac{{v}_{1}(t)}{1-u_b(t)} \\
     &\vdots   \\
             {v}_{d-1}'(t) &=  b(1-\tfrac{d-2}{d})\frac{{v}_{d-2}(t)}{1-u_b(t)}- b(1-\tfrac{d-1}{d})\frac{{v}_{d-1}(t)}{1-u_b(t)}\,, \\
     {v}_d'(t)&= b(1-\tfrac{d-1}{d})\frac{{v}_{d-1}(t)}{1-u_b(t)}\,.
    \end{cases}
\end{equation}

Moreover, for all $j=0,\ldots,d$ and $t \in [0,1)$, the solutions to the system in Eq.~\eqref{eq:system-overflow-bits} can be written explicitly in terms of $\lambda_b(t)$:
    \begin{equation}
    \label{eq:formula-prop-overflow-bits}
    v_j(t) = \binom{d}{j} e^{-(d-j)\lambda_b(t) /d} \int_0^{\lambda_b(t)}\tfrac{x^{b-1}}{(b-1)!} (e^{-x/d}-e^{-\lambda_b(t)/d})^j dx\,.\end{equation}

\end{restatable}

In order to prove this, it proves simpler to study $\tilde{V_j}(r)$, the number of full-buckets with $j$ bits after $r$ hops, instead of after $k$ insertions. So we change the time scale, each increment corresponding to a successful or unsuccessful access to a bucket, i.e. a hop.  
We can prove that these new random variables satisfy equations amenable to Wormald's Theorem:

\begin{equation*}\label{eq:V tilde}
    \mathbb{E}[\tilde{V_j}(r+1)-\tilde{V_j} (r) \,\big| \, \tilde{\mathcal{F}}_r] = (1-\tfrac{j-1}{d}) \tilde{V}_{j-1}(r)  -  (1-\tfrac{j}{d}) \tilde{V}_{j}(r)\,,\text{ for }j>0\,
\end{equation*}
and $\mathbb{E}[\tilde{V}_0(r+1)-\tilde{V}_0 (r) \,\big| \, \tilde{\mathcal{F}}_r] =   \tilde{U}_{b-1}(r) -  \tilde{V}_0(r)$, where $\tilde{U}_{b-1}(r)$ is the analog of $U_{b-1}$ but on the time-scale of hops $r$. 
Intuitively, the terms in $j/d$ (resp. $(j-1)/d$) account for the probability that the overflow bit chosen for this hop collides one of the $j$ (resp. $j-1$) bits set to 1 for the current bucket. More generally, we define new random variables $\tilde{U}_0,\tilde{U}_1,\ldots,\tilde{U}_b$ that denote the number of buckets with $0,1,\ldots,b$ spots occupied, at the new scale of time. They satisfy, for $j=0,\ldots,b-2$
\[
\mathbb{E}[\tilde{U}_{j+1}(r+1)-\tilde{U}_{j+1}(r)\,\big|\,\tilde{\mathcal{F}}_r] = \tilde{U}_j(r) - \tilde{U}_{j+1}(r)\,,
\]
while $\mathbb{E}[\tilde{U}_{0}(r+1)-\tilde{U}_{0}(r)\,\big|\,\tilde{\mathcal{F}}_r] = -\tilde{U}_0(r)$ and $\mathbb{E}[\tilde{U}_{b}(r+1)-\tilde{U}_{b}(r)\,\big|\,\tilde{\mathcal{F}}_r] = \tilde{U}_{b-1}(r)$.

These are related to the following system of ODE, again by a direct application of Wormald's Theorem:
\begin{equation*}
    \begin{cases}
            \tilde{u}_0'(s) &=  b\tilde{u}_{b-1}(s) - b \tilde{u}_0(s)\,, \\
            \tilde{u}_1'(s) &= b{\tilde{u}_{0}(s)} - b {\tilde{u}_1(s)} \\
     &\vdots   \\
             \tilde{u}_{b-1}'(s) &=  b{\tilde{u}_{b-2}(s)}- b{\tilde{u}_{b-1}(s)}\,, \\
     \tilde{u}_b'(s)&= b{\tilde{u}_{b-1}(s)}\,.
    \end{cases}
\end{equation*}
These functions are defined and continuously differentiable on $s\in [0,\infty)$ and Wormald's Theorem implies that $\tilde{U}_j(s)=n \tilde{u}_j(s/N)+o(n)\,.$ A key property linking the two scales of time is the following:
\begin{proposition}
    \label{prop:relationtion-time-scales}
    For all $j=0,\ldots,b$ and  $t\in [0,1)$,  we have $u_{j}(t)=\tilde{u}_j(\lambda_b(t)/b)$.
\end{proposition}
\begin{proof}[Proof sketch.]
    By definition $U_{j}(k) = \tilde{U}_{j}(R(k))$. Theorem~\ref{lambda-number-hops} tells us that $R(k)= \tfrac{N}{b}\lambda_b(k/N)+o(n)$ with high probability. Then $\tilde{U}_{j}(R(k)) = n \tilde{u}_{j}(\lambda_b(k/N)/b)+o(n)$ and $U_j(k)=n u_j(k/N)+o(n)$ both occur with high probability, the intersection of the events also occurs with high probability. This gives $u_j(k/N)\sim \tilde{u}_{j}(\lambda_b(k/N)/b)$, the rest follows by continuity.

\end{proof}

Back to the number of buckets with a given number of overflow its, we are led to the following system of ordinary differential equations, involving $\tilde{u}_{b-1}(t)$:
\begin{equation*}
    \begin{cases}
            \tilde{v}_0'(s) &=  b\tilde{u}_{b-1}(s) - b \tilde{v}_0(s)\,, \\
            \tilde{v}_1'(s) &= b(1-\tfrac{0}{d}){\tilde{v}_{0}(s)} - b(1-\tfrac{1}{d}) {\tilde{v}_1(s)} \\
     &\vdots   \\
             \tilde{v}_{d-1}'(s) &=  b(1-\tfrac{d-2}{d}){\tilde{v}_{d-2}(s)}- b(1-\tfrac{d-1}{d}){\tilde{v}_{d-1}(s)}\,, \\
     \tilde{v}_d'(s)&= b(1-\tfrac{d-1}{d}){\tilde{v}_{d-1}(s)}\,.
    \end{cases}
\end{equation*}

In this case the solutions are defined for $s\in [0,\infty)$ and Wormald's Differential Equation Method applies, yielding $\tilde{V}_j(s) =  n \tilde{v}_j(s/N) + o(n)$.  We are interested in $s=R(k)$, which are the actual insertion times. The relationship between $R(k)$ and $\lambda_b(t)$ was studied in Theorem~\ref{lambda-number-hops}. In particular, we recall that $R(k)=\tfrac{N}{b} \lambda_b(k/N)+o(n)$ with high probability. Thus $
    V_j(k)=\tilde{V}_j(R(k)) = n \tilde{v}_j(\lambda_b(k/N)/b)+o(n) $ with high probability. Set $v_j(t):=\tilde{v}_j(\lambda_b(t)/b)$.
Then
\[v_j'(t) = \tilde{v}_j'(\lambda_b(t)/b) \lambda_b'(t)/b = \frac{\tilde{v}_j'(\lambda_b(t)/b)}{1-u_b(t)}\,,
\] and we deduce, since $\tilde{u}_{b-1}(\lambda_b(t)/b)\equiv u_{b-1}(t)$, that the system in Eq.~\eqref{eq:system-overflow-bits} holds. \reviewer{See Annex~\ref{sec:D formula} for the proof.} It only remains to find manageable expressions for the $v_j$'s. The details of this inductive proof are left out of this extended abstract.

Recall that insertions and successful searches are not affected by the overflow mechanism, which can only interrupt sequences of probes for unsuccessful searches. Our result for unsuccessful searches is based on the observation that the probability of continuing the sequence of probes is well approximated by $p:=\sum_{j=0}^d v_j(t)  \tfrac{j}{d}$, where $v_j$ is the probability of hitting a full bucket having $j$ overflow bits set to $1$, and $j/d$ is the probability of hitting one of those bits. Thus, the number of hops strongly resembles a geometric law of parameter $p$. Following the formula in Eq.~\eqref{eq:formula-prop-overflow-bits} for the functions $v_j(t)$, we have the following result.

\begin{restatable}[Unsuccessful search with overflow bits]{theorem}{expectedvalueoverflow} \label{thm:unsuccessful-overflow}
    The expected number of hops for a unsuccessful search in a table having $k$ elements and using $d$ overflow bits is, with high probability and uniformly for $t=\frac{k}{N}\in [0,\theta)$:

   \reviewer{See Annex~\ref{sec:theorem overflow unsuccessful} for the full proof.}
\[
\left(1-F_{b,d}(\lambda_b(k/N))\right)^{-1} + o(1)\,,   \qquad 
F_{b,d}(\lambda):=\int_0^{\lambda}   \tfrac{x^{b-1}}{(b-1)!} e^{-(1-1/d)x} (e^{-x/d}-e^{-\lambda/d})  dx \,.
\]
Moreover, $F_{b,d}(\lambda) = P(b,\lambda)-(1-1/d)^{-b} e^{-\lambda/d} P(b,\lambda \cdot (1-1/d))$, where $P$ is the regularized lower-incomplete $\Gamma$-function defined in Section~\ref{sec:lambda}.

\end{restatable}

\section{Conclusion}
As Fig.~\ref{fig:experiments} shows, the results obtained in this article, although asymptotic, describe very precisely the behavior of modern hash tables that exploit SIMD instructions. The mathematical functions that enable the analysis are involved, implicitly defined, but easily estimated numerically thanks to the integral representations that we have proposed. From a methodological point of view, the use of variants of the differential equation method developed by Wormald has proved to be particularly well suited, as it produces advanced results with high probability, and gives a very precise description of the state of the table after $k$ insertions; this is why we only need one experimental trial in Fig.~\ref{fig:experiments}. For example, Theorem~\ref{thm:unsuccessful-overflow} can be used to compare the impact of adding more overflow bits, as depicted in Fig.~\ref{fig:different d}.

\begin{wrapfigure}{r}{0.5\textwidth}
  \begin{center}
  \begin{tikzpicture}[scale=.6]
\small
\begin{axis}[
    width=12cm,
    height=8cm,
    xlabel={Proportion of table occupancy},
    ylabel={Extra hops for unsuccessful search},
    legend style={at={(0.03,0.97)},anchor=north west},
    grid=both,
    xticklabel={
        \pgfmathparse{\tick*100}
        \pgfmathprintnumber{\pgfmathresult}\%
    },
    ylabel style={yshift=-1em},
    yticklabel style={font=\footnotesize,/pgf/number format/fixed,/pgf/number format/precision=3},
    xticklabel style={font=\footnotesize},
    scaled ticks=false
]

\addplot[
    blue, thick
]
coordinates {
(0.4375, 0.0024174427917344516)
 (0.44625, 0.002924552183021093)
 (0.45499999999999996, 0.0035173229031226286)
 (0.46374999999999994, 0.004206545745867463)
 (0.4724999999999999, 0.0050038603660047766)
 (0.4812499999999999, 0.0059217748423152795)
 (0.4899999999999999, 0.006973685212841493)
 (0.49874999999999986, 0.008173895971717922)
 (0.5074999999999998, 0.009537642654413607)
 (0.5162499999999999, 0.011081117767217697)
 (0.5249999999999998, 0.0128215014388745)
 (0.5337499999999997, 0.014776998290083127)
 (0.5424999999999998, 0.016966882134106666)
 (0.5512499999999998, 0.019411550244293044)
 (0.5599999999999997, 0.02213258905850959)
 (0.5687499999999996, 0.025152853344293114)
 (0.5774999999999997, 0.028496561031279677)
 (0.5862499999999997, 0.03218940614011472)
 (0.5949999999999996, 0.03625869251226418)
 (0.6037499999999996, 0.04073349138782357)
 (0.6124999999999996, 0.045644826306144824)
 (0.6212499999999996, 0.05102588933795032)
 (0.6299999999999996, 0.056912293323168474)
 (0.6387499999999995, 0.0633423656175513)
 (0.6474999999999995, 0.07035748988256749)
 (0.6562499999999996, 0.07800250373679764)
 (0.6649999999999995, 0.08632616168647352)
 (0.6737499999999994, 0.09538167474950489)
 (0.6824999999999994, 0.10522734068726257)
 (0.6912499999999995, 0.11592728190014577)
 (0.6999999999999994, 0.12755231200925365)
 (0.7087499999999993, 0.14018095717998416)
 (0.7174999999999994, 0.1539006646688846)
 (0.7262499999999994, 0.16880923933075662)
 (0.7349999999999993, 0.18501655950503385)
 (0.7437499999999992, 0.20264663762708848)
 (0.7524999999999993, 0.2218401092152812)
 (0.7612499999999993, 0.2427572581553248)
 (0.7699999999999992, 0.26558171868521663)
 (0.7787499999999992, 0.29052503838998284)
 (0.7874999999999992, 0.31783234650128994)
 (0.7962499999999992, 0.34778945469721617)
 (0.8049999999999992, 0.3807318335890296)
 (0.8137499999999991, 0.4170560725521268)
 (0.8224999999999991, 0.4572346671610815)
 (0.8312499999999992, 0.5018353242437505)
 (0.8399999999999991, 0.551546488558032)
 (0.848749999999999, 0.607211573644504)
 (0.857499999999999, 0.6698755832987617)
 (0.8662499999999991, 0.7408497147306552)
 (0.874999999999999, 0.8218026252701995)
 (0.8837499999999989, 0.9148922049710816)
 (0.892499999999999, 1.022960594657248)
 };
\addlegendentry{$b=16,d=0$}

\addplot[
    red,
    thick,
]
coordinates {
(0.4375, 0.00023553854654164309)
 (0.44625, 0.00028957790663419823)
 (0.45499999999999996, 0.0003539220878026547)
 (0.46374999999999994, 0.0004301268216203121)
 (0.4724999999999999, 0.0005199175114126013)
 (0.4812499999999999, 0.0006251999495306058)
 (0.4899999999999999, 0.0007480711958999784)
 (0.49874999999999986, 0.0008908306622728631)
 (0.5074999999999998, 0.0010559914666298828)
 (0.5162499999999999, 0.0012462921437476416)
 (0.5249999999999998, 0.001464708821021173)
 (0.5337499999999997, 0.0017144679932400053)
 (0.5424999999999998, 0.001999060056281278)
 (0.5512499999999998, 0.002322253787839947)
 (0.5599999999999997, 0.002688111993724346)
 (0.5687499999999996, 0.003101008571423266)
 (0.5774999999999997, 0.003565647279287659)
 (0.5862499999999997, 0.00408708254065848)
 (0.5949999999999996, 0.004670742658745277)
 (0.6037499999999996, 0.005322455871427034)
 (0.6124999999999996, 0.006048479737164265)
 (0.6212499999999996, 0.006855534416037122)
 (0.6299999999999996, 0.007750840496239044)
 (0.6387499999999995, 0.008742162119468641)
 (0.6474999999999995, 0.009837856282701207)
 (0.6562499999999996, 0.011046929343934291)
 (0.6649999999999995, 0.012379101942186744)
 (0.6737499999999994, 0.01384488376549582)
 (0.6824999999999994, 0.015455659875383108)
 (0.6912499999999995, 0.01722379063570048)
 (0.6999999999999994, 0.01916272771533814)
 (0.7087499999999993, 0.02128714916063163)
 (0.7174999999999994, 0.023613117194262997)
 (0.7262499999999994, 0.026158263232390767)
 (0.7349999999999993, 0.028942005673286076)
 (0.7437499999999992, 0.03198580736966609)
 (0.7524999999999993, 0.03531348144933929)
 (0.7612499999999993, 0.038951556426513115)
 (0.7699999999999992, 0.04292971453171489)
 (0.7787499999999992, 0.04728132113733337)
 (0.7874999999999992, 0.05204406843026363)
 (0.7962499999999992, 0.057260763602062545)
 (0.8049999999999992, 0.0629803015444513)
 (0.8137499999999991, 0.06925887546534464)
 (0.8224999999999991, 0.07616149764124018)
 (0.8312499999999992, 0.0837639292274056)
 (0.8399999999999991, 0.09215515658159686)
 (0.848749999999999, 0.10144060812351571)
 (0.857499999999999, 0.11174639037489888)
 (0.8662499999999991, 0.12322495109589024)
 (0.874999999999999, 0.13606277955714552)
 (0.8837499999999989, 0.15049107835723893)
 (0.892499999999999, 0.16680087724640125)
 };
\addlegendentry{$b=15,d=8$}

\addplot[
    green!50!black,
    thick,
]
coordinates {
(0.4375, 0.00016165517290008368)
 (0.44625, 0.0001966893504113234)
 (0.45499999999999996, 0.00023800402780745747)
 (0.46374999999999994, 0.0002864846515922822)
 (0.4724999999999999, 0.0003431059937177088)
 (0.4812499999999999, 0.0004089373302058165)
 (0.4899999999999999, 0.0004851477206799108)
 (0.49874999999999986, 0.0005730114189780398)
 (0.5074999999999998, 0.0006739134539038872)
 (0.5162499999999999, 0.0007893554285640107)
 (0.5249999999999998, 0.0009209615966672181)
 (0.5337499999999997, 0.0010704852846777542)
 (0.5424999999999998, 0.0012398157399193589)
 (0.5512499999999998, 0.0014309854967714548)
 (0.5599999999999997, 0.0016461783661865102)
 (0.5687499999999996, 0.0018877381681581323)
 (0.5774999999999997, 0.002158178342821498)
 (0.5862499999999997, 0.002460192593998991)
 (0.5949999999999996, 0.0027966667397310453)
 (0.6037499999999996, 0.0031706919683028873)
 (0.6124999999999996, 0.00358557972626457)
 (0.6212499999999996, 0.00404487849790737)
 (0.6299999999999996, 0.004552392774772325)
 (0.6387499999999995, 0.005112204560471535)
 (0.6474999999999995, 0.005728697812189369)
 (0.6562499999999996, 0.006406586287926053)
 (0.6649999999999995, 0.007150945350649445)
 (0.6737499999999994, 0.00796724838056057)
 (0.6824999999999994, 0.008861408569135084)
 (0.6912499999999995, 0.00983982701923391)
 (0.6999999999999994, 0.01090944826176178)
 (0.7087499999999993, 0.012077824530714838)
 (0.7174999999999994, 0.013353190427549068)
 (0.7262499999999994, 0.014744549969179588)
 (0.7349999999999993, 0.016261778473568647)
 (0.7437499999999992, 0.01791574232223818)
 (0.7524999999999993, 0.019718440389947835)
 (0.7612499999999993, 0.021683171902526124)
 (0.7699999999999992, 0.023824736749232536)
 (0.7787499999999992, 0.026159675940268783)
 (0.7874999999999992, 0.028706562109857225)
 (0.7962499999999992, 0.03148635292987649)
 (0.8049999999999992, 0.034522824320470526)
 (0.8137499999999991, 0.037843105864893783)
 (0.8224999999999991, 0.04147834851419193)
 (0.8312499999999992, 0.045464565498122146)
 (0.8399999999999991, 0.04984370287339745)
 (0.848749999999999, 0.054665018742689954)
 (0.857499999999999, 0.059986883722822716)
 (0.8662499999999991, 0.06587916605885345)
 (0.874999999999999, 0.07242644354491959)
 (0.8837499999999989, 0.07973240965034183)
 (0.892499999999999, 0.08792604619945853)
 };
\addlegendentry{$b=14,d=16$}
\end{axis}
\end{tikzpicture}
  \end{center}
  \caption{The expected number of additional hops for a unsuccessful search, for $n=2^{14}$. We use different values of $b$ and $d$ at constant header size, to compare the impact of the overflow bits. Observe that the capacities slightly differ since $b$ is not constant, but they do not vary much.\label{fig:different d}}
\end{wrapfigure}
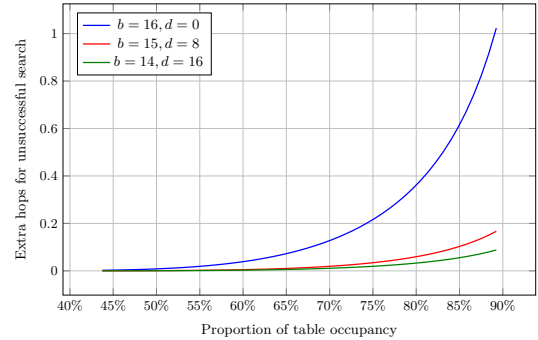

    In the results, there is a tradeoff between the error term and the speed of convergence in probability: the smaller the error term, the slower the convergence in probability. For each of the $o(n)$ terms in the results provided, we can prove that these are $O(n^{1/2+\varepsilon})$, for any fixed $\varepsilon>0$, with probability tending to one.  This follows from taking $\delta(n)=n^{-1/2+\varepsilon}$ when applying Theorem~\ref{thm:wormald-probabilistic-increment}. In consequence, the $o(1)$ terms are $O(n^{-1/2+\varepsilon})$. The probability depends on the choice of $\varepsilon>0$, but it is always $1-o(n^{-M})$ for every any $M>0$. See the proof of Theorem~\ref{lambda-number-hops}, for an example. This also applies to Section~\ref{sec:analysis-overflow} when composing.

As we already mentioned in the introduction, the engineers of each programming language have implemented their own algorithms, and it is difficult to cover everything. However, many of the variants we encountered can most likely be handled in the same way as what we have just done, perhaps with solutions that are harder to describe, but certainly still exploitable. For example, this method could be used to refine the analysis to count the number of element comparisons instead of the number of hops, sicne it is not difficult to add the  fingerprint collisions in the equations.
A notable exception to this statement concerns algorithms in which groups of $b$ elements are not arranged by bucket, but overlap within the table: such overlaps greatly complicate the analysis by introducing dependencies between the different buckets, and new ideas will be required to study them. 

Finally, we did not address deletions in this article. The main reason for this is that there is no commonly accepted model for the use of deletions in hash tables. In practice, it is often the case that elements are never deleted, but even when they are, defining typical usage scenarios is not an easy task. Some suggestions can be found in~\cite{larson1983analysis,MartinezNR25,PobleteV19}, for instance, but without real validation from practice~\footnote{In~\cite{bender2024}, there is also a scenario where they repeatedly delete a uniformly random element and then insert a new uniformly random element in an almost-full table. However, it is used to establish a lower bound via Yao's principle, not as a model of typical practical use of a hash table.}.  Such scenarios, using a simple probabilistic modeling, are most certainly amenable to analysis using the methodology of this article, though solving the systems might be difficult, but more advanced descriptions will likely require significant improvements of the method.

For example, consider the simple model introduced in~\cite{MartinezNR25} in which, at time $k$, an insertion occurs with probability $p\in(\tfrac{1}{2},1)$, and a deletion is performed (uniformly among all elements inside) with probability $1-p$. Observe that now it is possible to have buckets which are not full but have a number of overflow bits. Thus we are led to a system of ordinary differential equations with $(b+1)\times (d+1)$ unknown functions. If $U_{i,j}(k)$ corresponds to the number of buckets with $i$ elements and $j$ bits on, we expect a system for the fluid limit $u_{i,j}(t)=\lim U_{i,j}(\lfloor t b n\rfloor)/n$ by similar methods. Observe that we should have $\sum_{i,j} u_{i,j}(t)\cdot i = (2p-1)bt$, as this corresponds to the expected number of elements in the table. The variable $t$ must be constrained to $t \leq \theta /(2p-1)$, at which point the table attains the maximal load factor $\theta$ for the first time, with high probability. This system, which can be written explicitly, is more involved and we do not succeed in obtaining a simple solution, even though it can  easily be evaluated by numerical methods.

\bibliography{biblio}
\clearpage
\appendix

\section{Theorem~\ref{thm:wormald-probabilistic-increment}}

\subsection{Wormald's Theorem}

In this section we cite Wormald's original version of Wormald's Differential Equation method, as presented in \cite{wormald95}, for reference. In order to prove Theorem~\ref{thm:proportion-buckets-no-overflow} we could have used a result of Wormald's in almost its original form, but for the rest (e.g., Theorem~\ref{lambda-number-hops}) we need a stronger result that we present in Section~\ref{sec:proof-wormald-probabilistic-increment}.
    
\begin{theorem}[Adapted from {Wormald'95~\cite{wormald95}}]

    \label{thm:original-wormald}
    Let $\ell\geq 1$ be an integer. Consider discrete random processes $(Y_j)_{j=1}^\ell$ taking on real values, where $Y_j(k)$ is the value of the $j$-th process at time $k$. Suppose there exists $C>0$ such that $|Y_j(k)|\leq C \cdot N$ for all $k\leq N$ and $j\in \{1,\ldots,\ell\}$.

    Suppose also that for some $m=m(N)$, 
    \begin{enumerate}
        \item  there is a constant $\beta>0$ such that for all $k<m$ and $j\in [\ell]$,
        \[
        |Y_j(k+1)-Y_j(k)|\leq \beta\,;
        \]
        \item for all $j\in [\ell]$, and $k<m,$
        \[
        \mathbb{E}[Y_j(k+1)-Y_j(k) \,\,\big|\,\, \mathcal{F}_k] = f_j(k/N; Y_1(k)/N,\ldots,Y_b(k)/N) \,,
        \]
        where the functions $f_j $ are continuous and satisfy a Lipschitz condition on an open and connected domain $\mathcal{D}$;
        \item the initial vector $\boldsymbol{Y}(0)=(Y_1(0),\ldots,Y_\ell(0))$ satisfies $\|{\boldsymbol{Y}}(0)-N \boldsymbol{v}\|_\infty = o(N)$ for some fixed ${\boldsymbol{v}}=(v_1,\ldots,v_\ell)$ with $(0,v_1,\ldots,v_\ell)\in \mathcal{D}$.
    \end{enumerate}
        
        Let  $\boldsymbol{y}(t)=(y_1(t),\ldots,y_\ell(t))$ be the unique solution to the system of differential equations in Eq.~\eqref{eq:sys-ode}, with initial condition $y_1(0)=v_1,\ldots,y_\ell(0)=v_\ell$. Let $\sigma>0$ be the supremum of $t\geq 0$ such that $(s,\boldsymbol{y}(s))\in \mathcal{D}$ for all $s\in [0,t)$.

        Fix any $\theta \in (0,\sigma)$, then with probability tending to one we have $Y_j(k) = N y_j(k/N) + o(N)$, for some $o$-term that is actually uniform on $j$ and $k\leq \min(m,\theta N)$. 
            
\end{theorem}

Theorem~\ref{thm:original-wormald} can be extended to allow for differences bounded by $\beta=\beta(N)\to \infty$, provided that $\beta(N)$ does  not grow too fast. Moreover, it still holds if $|Y_j(k+1)-Y_j(k)|\leq \beta(N)$ occurs with high probability instead instead of always. These variants already appear in the  paper of Wormald~\cite{wormald95}. We have supposed that the initial condition is deterministic, as the initial state of our algorithm is deterministic. Wormald's Theorem, in its original form, allows for a random initial condition. It also allows for the \emph{trending hypothesis} $(\textcolor{darkgray}{\bf 2.})$ to be an approximation up to an uniform $o(1)$. The condition  $k\leq \theta N$ for $\theta<\sigma$ simplifies the statement somewhat. This can be lifted by being more precise regarding the distance to the border of $\mathcal{D}$ and the growth of the functions $f_j$. Hence Theorem~\ref{thm:original-wormald} is both an extension and a restriction of the statement of~\cite{wormald95}, but everything can be deduced directly from the work of Wormald. Our main tool, Theorem~\ref{thm:wormald-probabilistic-increment}, which is proved in the next section, is derived from Theorem~\ref{thm:original-wormald}.

\subsection{Proof of Theorem~\ref{thm:wormald-probabilistic-increment}}

\label{sec:proof-wormald-probabilistic-increment}

To prove Theorem~\ref{thm:wormald-probabilistic-increment}, we rely on two results from Warnke~\cite{warnke2019wormaldsdifferentialequationmethod}. First, we improve on Theorem~\ref{thm:original-wormald} by precising the error term.
\begin{theorem}[{Adapted Theorem~2 from \cite{warnke2019wormaldsdifferentialequationmethod}}]

    \label{thm:warnke-version}
            Let $\ell\geq 1$ be an integer. Consider discrete random processes $(Y_j)_{j=1}^\ell$ taking real values, where $Y_j(k)$ is the value of the $j$-th process at time $k$. 
    Suppose that there are functions $m:=m(N)$, $\beta=\beta(N)$, and $\delta:=\delta(N)$, such that $\lim_N\gamma(N) = \lim_N\delta(N)=0$ and 
    \begin{enumerate}
        \item  for all $k<m$ and $j\in [\ell]$,
        $
        |Y_j(k+1)-Y_j(k)|\leq \beta\,;
        $
        \item for all $j\in [\ell]$, and  $k<m,$
        \[
        \mathbb{E}[Y_j(k+1)-Y_j(k) \,\,\big|\,\, \mathcal{F}_k] = f_j(k/N; Y_1(k)/N,\ldots,Y_\ell(k)/N)\,,
        \]
        where the functions $f_j $ are continuous, uniformly bounded by $R>0$, and are $L$-Lipschitz on an open and connected domain $\mathcal{D}$, form some $L>0$. 
        \item The initial vector $\boldsymbol{Y}(0)=(Y_1(0),\ldots,Y_\ell(0))$ satisfies $\|{\boldsymbol{Y}}(0)-N \boldsymbol{v}\|_\infty \leq \delta\cdot N$ for some fixed ${\boldsymbol{v}}=(v_1,\ldots,v_\ell)$ with $(0,v_1,\ldots,v_\ell)\in \mathcal{D}$.
    \end{enumerate}    
        Let  $\boldsymbol{y}(t)=(y_1(t),\ldots,y_\ell(t))$ be the unique solution to the system of differential equations of Eq.~\eqref{eq:sys-ode}, with initial condition $y_1(0)=v_1,\ldots,y_\ell(0)=v_\ell$. Let $\sigma>0$ be the supremum of $t\geq 0$ such that the solution $(s,\boldsymbol{y}(s))\in \mathcal{D}$ for all $s\in [0,t)$ and fix $\theta \in (0,\sigma)$.

        Whenever $\delta\geq R/N$, with probability at least $1-2 \ell \exp(-N \delta^2 / (8\theta \beta^2)) $, we have $|Y_j(k) - N y_j(k/N) |\leq 3 e^{L \theta} \delta \cdot N$, for all $j \in [\ell]$ and $k\leq \min(m,\theta N)$. 

\end{theorem}

In order to prove Theorem~\ref{thm:wormald-probabilistic-increment} we still need Condition \textcolor{gray}{ \bf1} to be probabilistic, rather than deterministic, which is done in the following lemma from~\cite{warnke2019wormaldsdifferentialequationmethod}.

\begin{lemma}[{Lemma~9 from~\cite{warnke2019wormaldsdifferentialequationmethod}}]
    \label{lemma:conditioning-wormald}
    Consider events $(\mathcal{E}_i)_{i\leq I}$ with $\mathcal{E}_i \in \mathcal{F}_i$. Suppose that we  only verified Condition~\textcolor{darkgray}{ \bf1} and Condition~\textcolor{darkgray}{ \bf2} of Theorem~\ref{thm:warnke-version} under the asumption that $\mathcal{E}_i$ holds. Let $\mathcal{E}=\bigcap_{i\leq I} \mathcal{E}_i$. Then the probability of having $\mathcal{E}$ and the conclusion of Theorem~\ref{thm:warnke-version} {\bf{not}} holding is at most $2 \ell \exp(-N \delta^2/(8\theta \beta^2))$.
\end{lemma}

We are now ready to prove the theorem. 
\begin{proof}[Proof of {Theorem~\ref{thm:wormald-probabilistic-increment}}]
Let $\mathcal{E}_j$ be the event $|Y_{j+1}(k)-Y_j(k)|\leq \beta(N)$ for all $j\in [\ell]$ and $k\leq \theta N$ (here we can take $m=N$). Let $S=S(N,\delta,\beta)$ be the event that corresponds to the conclusions of Theorem~\ref{thm:warnke-version}.  Lemma~\ref{lemma:conditioning-wormald} tells us that if $\mathcal{E}=\bigcap_{j < \theta N} \mathcal{E}_j$, the probability $\Pr(\mathcal{E}\cap \overline{S})\leq 2 \ell \exp(-N \delta^2/(8\theta \beta^2))$, where $\overline{S}$ is the complement of $S$. Observe that $\Pr(\overline{\mathcal{E}}) \leq \ell\times\theta N\times \gamma(N)$ by applying the union-bound. We deduce that $\Pr(\overline{S})\leq \Pr(\overline{\mathcal{E}})+\Pr(\mathcal{E}\cap \overline{S}) \leq \ell \theta N \gamma(N) + 2 \ell \exp(-N \delta^2/(8\theta \beta^2))$ and the conclusion follows by taking the complement.
\end{proof}

\section{Proofs from Section~\ref{sec:no-overflow}}

\subsection{Proof of Theorem~\ref{thm:proportion-buckets-no-overflow}}
\label{sec:proof:thm:proportion-buckets-no-overflow}

In this section we complete the proof of the following theorem.

\thmproportionoccupied*

We begin by proving the following lemma, which we need in order to prove that $u_b(t)\leq t$ holds. This is required in order to assert that the solution to the system of ODEs exists on $[0,1)$, because of the denominators $1-u_b(t)$. Another consequence of the lemma is that $u_i(t)\leq 1$ holds for every $i$, and every $t$ where the solutions are defined.  Given these properties, we can verify that for every $\varepsilon \in (0,1)$ the system satisfies the hypothesis of the Picard-Lindel\"of Theorem for $t\in [0,1-\varepsilon)$.
\begin{lemma}
    \label{lemma:sols-are-positive}
    The solution $(u_0(t),u_1(t),\ldots,u_b(t))$ to the system of ODEs in Equation~\eqref{eq:sys-ode-occupied-buckets-no-overflow} with initial condition $u_0(0)=1,u_1(0)=0,\ldots,u_\bucketSize(0)=0$,  satisfies $u_i(t)>0$, for every $i>0$, and every $t>0$ where the solution can be extended.
\end{lemma}
\begin{proof}
It is immediate that $u_0(t)>0$ for $t>0$ as the system implies $u_0(t)=\exp\left(-b\int_0^t \frac{dt}{1-u_b(t)}\right)$.

For $i>0$, we claim that $u_i(t)>0$ for $t>0$ holds too. Suppose for the sake of contradiction that $t'>0$ were the smallest (by continuity) possible so that $u_i(t')=0$ for some $i>0$. We take the smallest $i$ without loss of generality, so that $u_{i-1}(t')>0$. Observe that $1-u_b(t')=u_0(t')+\ldots+u_{b-1}(t')>0$ by our choice of $t'$.  By the system, \[u_i'(t')=b\frac{ u_{i-1}(t')}{1-u_b(t')}-b\frac{u_i(t')}{1-u_b(t')} = b\frac{u'_{i-1}(t')}{1-u_b(t')} > 0\,,\] but this is a contradiction, as we must have $u_i'(t')\leq 0$ by the choice of $t'$.
\end{proof}

\begin{corollary}
    The solution $(u_0(t),u_1(t),\ldots,u_b(t))$ to the system of ODEs in Equation~\eqref{eq:sys-ode-occupied-buckets-no-overflow} with initial condition $u_0(0)=1,u_1(0)=0,\ldots,u_\bucketSize(0)=0$,  satisfies $u_i(t)\leq 1$ for every $i=0,1,\ldots,b$, and for every $t>0$ where the solution can be extended.
\end{corollary}
\begin{proof}
    Follows from $\sum_{i=0}^b u_i(t)=1$ and the non-negativity.
\end{proof}

\begin{corollary}
    The solution $(u_0(t),u_1(t),\ldots,u_b(t))$ to the system of ODEs in Equation~\eqref{eq:sys-ode-occupied-buckets-no-overflow} with initial condition $u_0(0)=1,u_1(0)=0,\ldots,u_\bucketSize(0)=0$,  satisfies $u_b(t)\leq t$ for every $t>0$ where the solution can be extended.
\end{corollary}
\begin{proof}
    From the system of ODEs in Eq.~\eqref{eq:sys-ode-occupied-buckets-no-overflow} and simple calculations we deduce that
    \[
    \sum_{i=0}^b i u_i^\prime(t)= b \frac{\sum_{i<b} u_i(t)}{1-u_b(t)} = b\,,
    \]
    as we know that $\sum_{i=0}^b u_i(t)=1$. We deduce that $\sum_i i u_i(t)=bt$. Then, since $u_i(t) > 0$ for $t>0$ and $i=0,\ldots,b$, we deduce that $b u_b(t) \leq b t$, thus $u_b(t)\leq t$.
\end{proof}

We now conclude the proof of Theorem~\ref{thm:proportion-buckets-no-overflow}.
Theorem~\ref{thm:wormald-probabilistic-increment} already gives the approximations of the $U_i(k)$'s by $nu_i(k/N)$, so
it remains to prove the expressions for the $u_i$'s are correct in the statement Theorem~\ref{thm:proportion-buckets-no-overflow}. Since the solution to the system of ODEs (ordinary differential equations) in Equation~\eqref{eq:sys-ode-occupied-buckets-no-overflow} is unique, given the initial conditions $u_0(0)=1,u_1(0)=0,\ldots,u_b(0)=0$, it suffices
to check that the proposed formula satisfies the system of ODEs. Note that $\lambda_b(t)$ is well defined, differentiable, and unique, due to the Implicit Function Theorem. 
Indeed, let $F(t,\lambda):=b-bt-\sum_{i=0}^{b} (b-i) \tfrac{\lambda^i}{i!} e^{-\lambda}$. Then  $\partial_\lambda F(t,\lambda)=-\sum_{i=0}^{b-1} (b-i-1) \tfrac{\lambda^i}{i!} e^{-\lambda}+\sum_{i=0}^{b} (b-i) \tfrac{\lambda^i}{i!} e^{-\lambda}=\sum_{i=0}^{b-1} \tfrac{\lambda^i}{i!} e^{-\lambda} \geq e^{-\lambda}>0$ for any $\lambda$.

Let $u_i(t)$ be the functions defined by the formulas at the end of the statement of Theorem~\ref{thm:proportion-buckets-no-overflow}, we show that they verify the system. They clearly satisfy the initial conditions. Observe that, for $0<i<b$, the functions defined by the formulas satisfy
\begin{equation}\label{eq:u prime}
u_i'(t) = \lambda_b'(t) \left( u_{i-1}(t) - u_i(t)\right),
\end{equation}
while $u_0'(t)=-u_0(t) \lambda_b'(t)$ and $u_b'(t)=-u_0'(t)-\ldots-u_{b-1}'(t)$. Hence, the system of ODEs is satisfied if we can also prove that $\lambda_b'(t)=\tfrac{b}{1-u_b(t)}$ and $u_b'(t)=\lambda_b'(t) u_{b-1}(t)$.

\begin{itemize}
    \item Observe that we have a telescoping sum $u_b'(t)=-u_0'(t)-\ldots-u_{b-1}'(t) = -\lambda_b'(t) u_b(t)$, hence it only remains to prove that $\lambda_b'(t)=\tfrac{b}{1-u_b(t)}$.

    \item Remark that by Eq.~\eqref{eq:u prime}, we have $\sum_{i=0}^b u_i'(t) i = \lambda_b'(t) \sum_{i=0}^{b-1}u_i(t)$.

    \item Finally, using basic computations and the identity above, the equation defining $\lambda_b(t)$ is equivalent to $    b-b t = b - \sum_{i=0}^b u_i(t) i$, since $u_b(t):=1-u_0(t)-\ldots-u_{b-1}(t)$.
\end{itemize}
    Differentiating the equation $    b-b t = b - \sum_{i=0}^b u_i(t) i$ yields
    \begin{align*}
        -b &= -\sum_{i=0}^b u_i'(t) i \\
        &= -\lambda_b'(t) \sum_{i=0}^{b-1}u_i(t) \\ 
        &= -\lambda_b'(t) \times (1-u_b(t))\,,
    \end{align*}
    completing the proof of Theorem~\ref{thm:proportion-buckets-no-overflow}
    .

\section{Proofs from Section~\ref{sec:lambda}}
\label{sec:proofpropertieslambda}

\propertieslambda*
\begin{proof}
First, the inequality $u_b(t)\leq t$ implies that $\lambda_b'(t) = \frac{b}{1-u_b(t)} \leq \frac{b}{1-t}$ and integrating we obtain $\lambda_b(t)\leq b \log\left(\frac{1}{1-t}\right)$ for all $t \in [0,1)$.

Next, observe that $p(\lambda_b(t))e^{-\lambda_b(t)}=b-bt$ for $p(x):=\sum_{i\leq b} (b-i) \frac{x^i}{i!}$. Clearly,  the coefficients of the polynomial $p(x)$ being positive, we have $p(0)\leq p(x)$ for every $x\geq 0$. It follows that $b e^{-\lambda_b(t)} \leq b - bt$. Taking logarithms we obtain that $\log \left(\frac{1}{1-t}\right)\leq \lambda_b(t)$  for all $t \in [0,1)$. 

Finally, we move on to the asymptotics as $t\to 0^+$ and $t\to 1^-$. Observe that $\lambda_b(t)\geq 1$ for $t\geq 1-e^{-1}$ due to the previous inequality. Then we can bound simply $p(\lambda_b(t))\leq b^2 (\lambda_b(t))^{b-1} $. Taking logarithms we obtain $\log b + \log (1-t) \leq -\lambda_b(t) +2\log b+(b-1)\log \lambda_b(t)$. Using $\lambda_b(t)\leq b \log\left(\frac{1}{1-t}\right)$, we deduce
\[
\lambda_b(t) \leq \log\left(\tfrac{1}{1-t}\right) + (b-1)\log\log \left(\tfrac{1}{1-t}\right) + b\log b\,,
\]
for $t\geq 1-e^{-1}$. Along with $\log \left(\frac{1}{1-t}\right)\leq \lambda_b(t)$ , this proves the asymptotic as $t\to 1^-$.

On the other hand $\lambda_b'(t)\geq b$ and so $\lambda_b(t)\geq b t$. Since $b t \sim b \log(\tfrac{1}{1-t})$ as $t\to 0^+$, the other equivalent for $t\to 0^+$ follows because we have seen that $\lambda_b(t)\leq b \log\left(\frac{1}{1-t}\right)$ holds.
\end{proof}

\gammaincomplete*

\begin{proof}
    Note that $1-u_b(t) = \sum_{i=0}^{b-1} u_i(t) = \sum_{i=0}^{b-1} \tfrac{\lambda_b(t)^i}{i!}e^{-\lambda_b(t)} = Q(b,\lambda_b(t))$, the last equality being a classical relation $Q(b,\omega) = \Pr({\tt Poiss}(\omega)< b)$ for $b=1,2,\ldots$ where ${\tt Poiss}(\omega)$ denotes a Poisson random variable of parameter $\omega$. The rest follows from Lemma~\ref{lemma:lambda-integral}.
\end{proof}

\section{Proofs from Section~\ref{sec:analysis-overflow}}

\subsection{Proving the formula}\label{sec:D formula}

We will need the following lemma to prove the formula for $v_j(t)$.

\begin{restatable}{lemma}{integralinduction}
    
    \label{lemma:integral-induction}
    Let $F(x,y)$,  $F\colon [0,\infty)\times [0,\infty)\to \mathbb{R}$, be continuous on $(x,y)$, and  $y\mapsto F(x,y)$ continuously differentiable on $y$ for almost every $x$, then
    \[
    \tfrac{d}{dy} \left(y\mapsto \int_0^y F(x,y)dx\right) = F(y,y) + \int_0^y \tfrac{\partial}{\partial y}F (x,y) dx\,.
    \]
    In particular, for any integers $b,d,s\geq 1$,
    \[
\tfrac{d}{dy}\left(\int_0^{y}    x^{b-1} (e^{-x/d} - e^{-y/d})^{s} dx\right) =  \tfrac{s}{d} e^{-y/d}  \int_0^{y}    x^{b-1} (e^{-x/d} - e^{-y/d})^{s-1} dx\,.
\]

\end{restatable}

\begin{proof}
    The proof is just a direct verification: we show that the primitives coincide, hence so do the functions. If we integrate the right-hand side from $0$ to $a$ we obtain:
\[
\int_0^a F(y,y)dy + \int_0^a \left(\int_0^y \tfrac{\partial}{\partial y}F(x,y)dx\right)dy\,.
\]
Reversing the order of integration on the second term by Fubini's Theorem (as the partial derivative $\tfrac{\partial}{\partial y}F(x,y)$ is integrable on $[0,a]\times [0,y]$) we obtain
\[
\int_0^a F(y,y)dy + \int_0^y \left( \int_x^a\tfrac{\partial}{\partial y}F(x,y)dy\right)dx = \int_0^a F(y,y)dy + \int_0^a \left( F(x,a)-F(x,x)\right)dx\,.
\]
Thus we simplify
\[
\int_0^a F(y,y)dy + \int_0^y \left( \int_0^a\tfrac{\partial}{\partial y}F(x,y)dy\right)dx = \int_0^a F(y,y)dy + \int_0^a F(x,a)dx - \int_0^a F(x,x)dx\,,
\]
which is exactly  $a\mapsto \int_0^a F(x,a)dx$. Thus the proof is complete.
\end{proof}

\paragraph*{Proof of the formula}
We prove inductively that the result holds for all $j\leq d$.

\noindent$\bullet$ Base case ($j=0$). Recall that $v_0$ can be defined by $v_0^\prime(t) = b\frac{u_{b-1}(t)}{1-u_b(t)} - b \tfrac{v_0(t)}{1-u_b(t)}$.  Using the product rule and $\lambda_b^\prime(t)=b\frac{1}{1-u_b(t)}$ (see Lemma~\ref{lemma:lambda-integral}) we obtain \[\left(v_0(t)e^{\lambda_b(t)}\right)^\prime = u_{b-1}(t) \lambda_b^\prime(t) e^{\lambda_b(t)} = \tfrac{\lambda_b(t)^{b-1}}{(b-1)!} \lambda_b^\prime(t)\,,\] where we have used that $u_{b-1}(t)=\tfrac{\lambda_b(t)^{b-1}}{(b-1)!} e^{-\lambda_b(t)}$. 
Integrating, with the change of variables $x=\lambda_b(t)$, gives
 \[v_0(t)e^{\lambda_b(t)} = \int_0^t \tfrac{\lambda_b(t)^{b-1}}{(b-1)!} \lambda_b^\prime(t) dt = \int_0^{\lambda_b(t)} \frac{x^{b-1}}{(b-1)!} dx \,,\] 
which proves the base case.

\noindent$\bullet$ Inductive step:  assume the result holds for all $j^\prime$, $0\leq j\prime<j\leq d$. We prove that it holds for $j$. Recall that
\[
v_j^\prime(t) = b \left(1-\tfrac{j-1}{d}\right) \frac{v_{j-1}(t)}{1-u_b(t)} - b \left(1-\tfrac{j}{d}\right) \frac{v_{s}(t)}{1-u_b(t)}\,,
\]
which is also valid for $j=d$, in which case the latter term is $0$ due to the factor $\left(1-\tfrac{j}{d}\right)$.
Thus 
\[
\left(v_j(t) e^{(1-\tfrac{j}{d})\lambda_b(t)}\right)^\prime  =
 b \left(1-\tfrac{j-1}{d}\right) \frac{v_{j-1}(t)}{1-u_b(t)} e^{(1-\tfrac{j}{d})\lambda_b(t)} =
 \left(1-\tfrac{j-1}{d}\right)\lambda_b^\prime(t) e^{(1-\tfrac{j}{d})\lambda_b(t)} v_{j-1}(t)\,.
\]
By the induction hypothesis $v_{j-1}(t)$ satisfies the proposed formula, thus
 {\footnotesize
\begin{align*}
    \left(v_j(t) e^{(1-\tfrac{j}{d})\lambda_b(t)}\right)^\prime 
&=
  \left(1-\tfrac{j-1}{d}\right) \lambda_b^\prime(t) e^{(1-\tfrac{j}{d})\lambda_b(t)} \cdot \binom{d}{j-1}  e^{-(d-(j-1)) \lambda_b(t)/d}\int_0^{\lambda_b(t)}   \tfrac{x^{b-1}}{(b-1)!} (e^{-x/d}-e^{-\lambda_b(t)/d})^{j-1} dx\,, \\
  &= \left(\tfrac{d-j+1}{d}\right) \binom{d}{j-1} e^{-\tfrac{1}{d}\lambda_b(t)}\lambda_b^\prime(t) \times \int_0^{\lambda_b(t)}   \tfrac{x^{b-1}}{(b-1)!} (e^{-x/d}-e^{-\lambda_b(t)/d})^{j-1} dx\,.
\end{align*}
}
Integrating we find
\[
v_j(t) e^{(1-\tfrac{j}{d})\lambda_b(t)} = \left(\tfrac{d-j+1}{d}\right) \binom{d}{j-1}  \int_0^t \left( e^{-\tfrac{1}{d}\lambda_b(t)}\lambda_b^\prime(t) \times \int_0^{\lambda_b(t)}   \tfrac{x^{b-1}}{(b-1)!} (e^{-x/d}-e^{-\lambda_b(t)/d})^{j-1} dx\right)dt\,.
\]
And changing variables $y=\lambda_b(t)$ on the outer integral,
\[
v_j(t) e^{(1-\tfrac{j}{d})\lambda_b(t)} = \left(\tfrac{d-j+1}{d}\right) \binom{d}{j-1} \int_0^{\lambda_b(t)} \left( e^{-\tfrac{y}{d}}\int_0^{y}   \tfrac{x^{b-1}}{(b-1)!} (e^{-x/d}-e^{-y/d})^{j-1} dx\right)dy\,.
\]

Integrating the result in Lemma~\ref{lemma:integral-induction}, we have the identity
\[
 \int_0^{\lambda_b(t)} \left( e^{-y/d}  \int_0^{y}    x^{b-1} (e^{-x/d} - e^{-y/d})^{j-1}dx\right) dy
  = \tfrac{d}{j}   \int_0^{\lambda_b(t)}    x^{b-1} (e^{-x/d} - e^{-\lambda_b(t)/d})^{j} dx
\]
which completes the proof for all $j\leq d$, since $\frac{d-j+1}{d} \frac{d}{j} \binom{d}{j-1} = \binom{d}{j}$ for $j>0$.

\begin{remark}
    Equivalently, one could prove this in the time of ``hops'' and then retrieve the result by composing with $\lambda_b(t)/b$.
\end{remark}

\subsection{Proof of Theorem~\ref{thm:unsuccessful-overflow}}
\label{sec:theorem overflow unsuccessful}

\expectedvalueoverflow*

\begin{proof}
We recall that the unsuccessful search, after $k$ insertions, is a geometric random variable
with parameter $P(k)=\sum_{j=0}^d \tfrac{V_j(k)}{n} \tfrac{j}{d}$. Hence, the expected value is $\tfrac{1}{1-P(k)}$. From Theorem~\ref{thm:proportion-buckets-overflow} we deduce that $P(k)=\sum_{j=0}^d v_j(t)  \tfrac{j}{d} + o(1)$ where $t=k/N$. Let $p(t):=\sum_{j=0}^d v_j(t)  \tfrac{j}{d} $, then it follows that $\tfrac{1}{1-P(k)}=\tfrac{1}{1-p(t)}+o(1)$. This approximation is valid, because  we  have that $P(k)\leq U_b(k)/n \leq \theta$ and $p(t)\leq u_b(t) \leq \theta$ when $k\leq \theta N$ and $t\leq \theta$, thus we are on the region where $x\mapsto \tfrac{1}{1-x}$ has a bounded derivative. Let us now calculate $p(t)$.

As a consequence of Theorem~\ref{thm:proportion-buckets-overflow} and the the Binomial Theorem, we obtain the following formula:
\begin{equation}
\label{gf-v}
\sum_{j=0}^d v_j(t) z^j = \int_0^{\lambda_b(t)}   \frac{x^{b-1}}{(b-1)!} \cdot \left(e^{-\lambda_b(t)/d} + z\cdot (e^{-x/d}-e^{-\lambda_b(t)/d}) \right)^d  dx\,.
\end{equation}
Differentiating in $z$ and letting $z=1$ we obtain, after diving by $d$
    \[
p(t)=\sum_{j=0}^d v_j(t) \tfrac{j}{d}=\int_0^{\lambda_b(t)}   \tfrac{x^{b-1}}{(b-1)!} e^{-(1-1/d)x} (e^{-x/d}-e^{-\lambda_b(t)/d})  dx \,.
\]
To deduce the given formula we expand
\[
p(t)=\int_0^{\lambda_b(t)}   \tfrac{x^{b-1}}{(b-1)!} e^{-x}  dx  - e^{-\lambda_b(t)/d}\int_0^{\lambda_b(t)}   \tfrac{x^{b-1}}{(b-1)!} e^{-(1-1/d)x}   dx\,.
\]
The first term is simply $P(b,\lambda_b(t))$. For the second we set $u=(1-1/d)x$ and we obtain 
\[
\int_0^{\lambda_b(t)}   \tfrac{x^{b-1}}{(b-1)!} e^{-(1-1/d)x}   dx = (1-1/d)^{-b} P(b,(1-1/d)\lambda_b(t))\,,
\]
thus proving the formula.
\end{proof}
\end{document}